\documentclass[letterpaper,11pt]{article}
\usepackage{microtype}
\usepackage{mathtools}
\usepackage{amsmath}
\usepackage{amsthm}
\usepackage{graphicx} 
\usepackage{xcolor}
\usepackage{amsfonts}
\usepackage{bm}
\usepackage{caption}
\usepackage{pifont}
\usepackage{dsfont}
\usepackage[linesnumbered]{algorithm2e}
\RestyleAlgo{ruled}
\usepackage{subcaption}
\usepackage[commandnameprefix=always]{changes}
\usepackage{xparse}
\usepackage{wrapfig}
\usepackage{comment}
\usepackage{fullpage}
\usepackage{amssymb}
\usepackage{esvect}
\usepackage{thmtools}
\usepackage{thm-restate}
\usepackage{import}
\usepackage{framed}
\usepackage{cite}
\usepackage{enumitem}
\usepackage{hyperref}
\usepackage{mathrsfs}
\usepackage[capitalize]{cleveref}
\definecolor{myblue}{HTML}{0088cc}
\hypersetup{
    colorlinks=true,
    linkcolor=black,
    citecolor=myblue,
    filecolor=black,
    urlcolor=myblue,
}

\DeclareUrlCommand{\path}{}
\declaretheorem[name=Theorem,numberwithin=section]{theorem}
\declaretheorem[name=Lemma,sibling=theorem]{lemma}
\usetikzlibrary{calc}

\DeclareMathOperator{\dist}{dist}
\DeclareMathOperator{\edist}{edist}

\NewDocumentCommand{\inlabels}{o}{%
  \ensuremath{%
    \Sigma^{\operatorname{in}}%
    \IfValueT{#1}{_{#1}}%
  }%
}

\NewDocumentCommand{\outlabels}{o}{%
  \ensuremath{%
    \Sigma^{\operatorname{out}}%
    \IfValueT{#1}{_{#1}}%
  }%
}

\newcommand{\ball}[3]{\ensuremath B^{#1}_{#2}(#3)}

\newcommand{\rsigma}[3]{\ensuremath \sigma^{\operatorname{#2}}_{#1}\upharpoonright_{#3}}

\newcommand{\ssigma}[2]{\ensuremath \sigma^{\operatorname{#2}}_{#1}}

\newcommand{\alg}[1]{\ensuremath \mathcal{A}_{#1}}

\newcommand{\MM}{\mathcal{M}}

\newcommand{\CC}{\mathcal{C}}

\newcommand{\ngadget}[2]{\ensuremath \mathsf{N}_{#1,#2}}
\newcommand{\egadget}[2]{\ensuremath \mathsf{E}_{#1,#2}}
\newcommand{\decode}[2]{\ensuremath \mathsf{decode}^{-1}(#1,#2)}

\newcommand{\dept}{dependent\xspace}
\newcommand{\local}{LOCAL\xspace}
\newcommand{\qlocal}{quantum-\local}
\newcommand{\detlocal}{det-\local}
\newcommand{\randlocal}{rand-\local} 
\newcommand{\slocal}{S\local}
\newcommand{\olocal}{online-\local} 
\newcommand{\findept}{finitely-\dept}
\newcommand{\nonsign}{non-signaling\xspace}

\newcommand{\id}{\mathsf{id}}

\DeclareMathOperator{\diam}{diam}
\DeclareMathOperator{\View}{View}
\DeclareMathOperator{\instance}{I}

\def\gn{\ensuremath \text{gball}}
\def\AA{\ensuremath \mathcal{A}}

\newtheorem{corollary}[theorem]{Corollary}
\theoremstyle{definition}
\newtheorem{definition}[theorem]{Definition}
\theoremstyle{remark}
\newtheorem*{remark}{Remark}

\DeclareMathOperator{\poly}{poly} 

\newcommand{\SWAP}{\operatorname{SWAP}}
\newcommand{\LIFT}{\operatorname{Lift}}

\newcommand{\outcome}{\operatorname{O}}

\DeclareMathOperator{\inpt}{in}
\DeclareMathOperator{\oupt}{out}

\begin{document}

\begin{flushleft}
    \huge\bf
    It does not matter how you define \\ locally checkable labelings
\end{flushleft}
\smallskip
 \begin{flushleft}
     \setlength{\parskip}{3pt}

     \textbf{Antonio Cruciani} · Aalto University

     \textbf{Avinandan Das} · Aalto University

     \textbf{Alesya Raevskaya} · Aalto University

     \textbf{Jukka Suomela} · Aalto University
 \end{flushleft}
\smallskip
\paragraph{Abstract.}
Locally checkable labeling problems (LCLs) form the foundation of the modern theory of distributed graph algorithms. First introduced in the seminal paper by Naor and Stockmeyer [STOC 1993], these are graph problems that can be described by listing a \emph{finite} set of valid local neighborhoods. This seemingly simple definition strikes a careful balance between two objectives: they are a family of problems that is broad enough so that it captures numerous problems that are of interest to researchers working in this field, yet restrictive enough so that it is possible to prove strong theorems that hold for all LCL problems. In particular, the distributed complexity landscape of LCL problems is now very well understood.

Yet what we are seeing more and more often is that there are LCL problems with unexpected, counterintuitive properties, that we have never seen in ``natural'' local graph problems, for example:
\begin{itemize}[noitemsep]
    \item There is an LCL problem that admits a distributed quantum advantage [SODA 2026].
    \item There is an LCL problem that benefits from shared randomness and from shared quantum state [ICALP 2026].
    \item There is an LCL problem whose complexity in the LOCAL model depends on whether state transitions are restricted to computable functions.
    \item We can construct LCL problems with unnatural round complexities such as $\Theta(\log^{123.45} n)$ and $\Theta(n^{0.12345})$ [STOC 2018].
\end{itemize}
Furthermore, many questions related to the distributed complexity of LCL problems are undecidable [STOC 1993, PODC 2017]. Overall, the LCL problems seem to be a poor proxy for the family of natural local graph problems. With large volumes of research on LCL problems, the following questions are getting ever more pressing: \emph{Are we studying the right problem family?} Maybe the above counterintuitive features are mere artifacts of the specific definition of LCLs introduced by Naor and Stockmeyer, and they disappear if we slightly restrict the family of LCL problems?

In this work we show that the family of LCL problems is \emph{extremely robust} to variations. We present a very restricted family of locally checkable problems (essentially, the ``node-edge checkable'' formalism familiar from round elimination, restricted to regular unlabeled graphs); most importantly, such problems cannot directly refer to e.g.\ the existence of short cycles. We show that one can translate between the two formalisms (there are local reductions in both directions that only need access to a symmetry-breaking oracle, and hence the overhead is at most an additive $O(\log^* n)$ rounds in the LOCAL model). In particular, all counterintuitive properties listed above hold also for restricted LCLs.

\thispagestyle{empty}
\setcounter{page}{0}
\newpage

\section{Introduction}

The modern complexity theory of distributed graph algorithms largely builds on the foundation of \emph{locally checkable labeling} problems (LCLs), first introduced by Naor and Stockmeyer \cite{naor-stockmeyer-1995-what-can-be-computed-locally} in their Dijkstra Prize-winning work.
Thanks to a large international research effort over the past ten years, spanning a number of research groups around the world, we now have a thorough understanding of the complexity landscape of LCL problems across numerous models of distributed computing \cite{
    brandt-fischer-etal-2016-a-lower-bound-for-the,
    fischer-ghaffari-2017-sublogarithmic-distributed,
    brandt-hirvonen-etal-2017-lcl-problems-on-grids,
    ghaffari-harris-kuhn-2018-on-derandomizing-local,
    balliu-hirvonen-etal-2018-new-classes-of-distributed,
    balliu-brandt-etal-2019-the-distributed-complexity-of,
    chang-pettie-2019-a-time-hierarchy-theorem-for-the,
    chang-kopelowitz-pettie-2019-an-exponential-separation,
    chang-2020-the-complexity-landscape-of-distributed,
    rozhon-ghaffari-2020-polylogarithmic-time-deterministic,
    balliu-brandt-etal-2020-how-much-does-randomness-help,
    balliu-brandt-etal-2021-almost-global-problems-in-the,
    balliu-censor-hillel-etal-2021-locally-checkable,
    balliu-brandt-etal-2022-efficient-classification-of,
    grunau-rozhon-brandt-2022-the-landscape-of-distributed,
    akbari-eslami-etal-2023-locality-in-online-dynamic,
    chang-studeny-suomela-2023-distributed-graph-problems,
    dhar-kujawa-etal-2024-local-problems-in-trees-across-a,
    akbari-coiteux-roy-etal-2025-online-locality-meets
}.
But are LCL problems the \emph{right} problem family to study?

\paragraph{What are LCL problems?}

LCL problems are, put simply, graph problems that can be described by listing a \emph{finite} set of valid local neighborhoods \cite{naor-stockmeyer-1995-what-can-be-computed-locally}. These are constraint-satisfaction problems defined on bounded-degree graphs; we have a finite set of input labels, a finite set of output labels, and a finite locally checkable constraint that relates the local graph structure to valid input--output combinations.

Numerous problems that have been studied in our field are LCL problems (when restricted to bounded-degree graphs); examples include vertex coloring, edge coloring, maximal independent set, maximal matching, and sinkless orientations. Such problems are particularly natural to study in a distributed setting, as a solution is easy to \emph{verify} with a distributed algorithm---hence they can also be seen as a natural distributed analog of the class NP (or, more precisely, FNP): these problems might be hard to solve, but solutions are at least easy to verify.

\paragraph{Why do we study LCL problems?}

Intuitively, LCL problems aim to strike a balance between two objectives:
\begin{enumerate}
    \item They are a family of problems that is broad enough so that it captures numerous problems that are of interest to researchers working in this field.
    \item They are restrictive enough so that it is possible to prove strong, useful theorems that hold for all LCL problems.
\end{enumerate}
LCL problems can be seen as one possible mathematically precise formalization of the family of ``simple, natural graph problems.'' LCL problems exclude numerous pathological graph problems that have an inherently global definition and hence are not that interesting from the perspective of, say, the LOCAL model of distributed computing. Ideally, we could prove strong results about LCL problems, and such results would be widely applicable to natural graph problems.

\paragraph{Success stories.}

This intuitive idea has been largely successful. For example, many natural graph problems are such that their round complexity (in bounded-degree graphs with $n$ nodes) in the usual deterministic LOCAL model is either $O(\log^* n)$ or $\Omega(\log n)$, but it is rare to see any natural problems that would have round complexity strictly between $\omega(\log^* n)$ and $o(\log n)$. Work on LCL problems has provided a clear explanation of this phenomenon: it can be shown \cite{chang-kopelowitz-pettie-2016-an-exponential-separation} that there is no LCL problem with round complexity between $\omega(\log^* n)$ and $o(\log n)$. Moreover, the proof is constructive; as soon as you can design an $o(\log n)$-round algorithm for any LCL problem, we can immediately speed it up to $O(\log^* n)$ rounds.

\paragraph{Recent counterintuitive findings.}

However, we are also seeing more and more findings where results on LCL problems disagree with our understanding of natural distributed graph problems:
\begin{itemize}
    \item There is an artificial LCL problem that admits a distributed quantum advantage in the LOCAL model \cite{balliu-casagrande-etal-2026-distributed-quantum}, but there is no known natural problem that would separate the LOCAL model from the quantum-LOCAL model.
    \item There is an artificial LCL problem that benefits from shared randomness and from shared quantum state \cite{balliu-ghaffari-etal-2025-shared-randomness-helps-with}, and again no natural LCL problem is known where shared randomness helps.
    \item There is an artificial LCL problem whose complexity in the LOCAL model depends on whether state transitions are restricted to computable functions \cite{local-computability-2026}, and again no such natural problem is known.
    \item We can construct artificial LCL problems with round complexities such as $\Theta(\log^{123.45} n)$ and $\Theta(n^{0.12345})$ \cite{balliu-hirvonen-etal-2018-new-classes-of-distributed}, which we do not see in any natural graph problem.
    \item Many questions related to the distributed complexity of LCL problems are undecidable; we can create artificial LCL problems that essentially encode the halting problem of Turing machines in their definitions \cite{naor-stockmeyer-1995-what-can-be-computed-locally,brandt-hirvonen-etal-2017-lcl-problems-on-grids}.
\end{itemize}
All of these examples seem to point to the conclusion that the family of LCL problems is \emph{too broad}, and it contains not only natural graph problems but also numerous artificial problems with features unlike anything we see in typical problems of interest. The following questions are becoming more pressing over time:
\begin{quote}
    \em
    Are we perhaps studying the wrong problem family?

    Are LCLs as a problem family way too broad?

    Could it be the case that there is some natural family of ``restricted LCLs'' that still contains most practically relevant distributed graph problems, yet excludes most of the artificial problems mentioned above?
\end{quote}

\paragraph{Key features that artificial problems exploit.}

In essence, all artificial LCL problems that we have seen fundamentally rely on the following elements of the definition of LCLs:
\begin{enumerate}
    \item The problem definition can refer to the \emph{local graph structure}. For example, the problem definition can capture the idea that we need to solve a certain problem if the graph is locally grid-like.
    \item The problem definition can refer to \emph{local inputs} (or some local feature that can trivially encode inputs, e.g., degrees of the nodes).
\end{enumerate}
At the same time, the usual distributed problems such as vertex coloring, maximal independent set, and sinkless orientation do not need to refer to the local graph structure or local inputs. Of course, to define a problem such as vertex coloring, we need to be able to capture the notion of \emph{adjacency}, but what we mean is that the problem definition does not need to refer to, say, the existence of $4$-cycles in a local neighborhood.

\paragraph{RE formalism---a concrete candidate definition of restricted LCLs.}

Let us now try to lay out a concrete candidate definition of ``restricted LCLs.'' We intentionally try to overshoot so that the definition is rather \emph{too restricted}, and see what happens then. We try to make sure that there is no way to directly refer to the local graph structure or inputs, or anything that could be used to directly encode inputs. We will follow the same formalism that has been successfully used in the context of the \emph{round elimination} technique \cite{brandt-2019-an-automatic-speedup-theorem-for,olivetti-2020-brief-announcement-round-eliminator-a,olivetti-2025-round-eliminator-a-tool-for-automatic}; we call this the \emph{RE formalism}.

We will assume that the input graph is an unlabeled $3$-regular simple graph and our task is to \emph{label half-edges}. We have a \emph{node constraint} that indicates which half-edge labels are valid combinations incident to a node (this is a set of $3$-element multisets) and an \emph{edge constraint} that indicates which half-edge labels are valid combinations at the endpoints of an edge (this is a set of $2$-element multisets).

With some thinking, it is not hard to see that this definition still allows us to encode problems essentially equivalent to our familiar distributed graph problems through simple local reductions. For example, the task of $4$-vertex coloring can be captured so that the set of half-edge labels is $\{1,2,3,4\}$, the node constraint specifies that all half-edges incident to a node have the same color (representing the color of the node), and the edge constraint specifies that the half-edge labels along a single edge must have distinct colors (representing the requirement that each edge must be properly colored). We can also use half-edge label pairs to represent e.g.\ edge orientations, and this way encode problems similar to the \emph{sinkless orientation} problem.

However, it is not at all clear if we can encode any of the artificial problems mentioned above in this formalism. Note that in the RE formalism, the \emph{feasibility of solutions is preserved under arbitrary lifts} (covering maps), which does not hold in the usual LCL formalism. In particular, we can take a short cycle and eliminate it with an appropriate lift, yet the same solution remains valid. Furthermore, after taking appropriate lifts, the local graph structure looks like a $3$-regular tree, and hence there does not seem to be any way to encode inputs.

Our original intuition was that the RE formalism is, in essence, at least as restrictive as LCL problems in regular unlabeled trees (where most of the counterintuitive artificial problems do not exist), and we set out to formally prove this intuition. Furthermore, we expected that we would eventually need to \emph{strengthen} the definition to capture natural problems such as list coloring that seem to genuinely require the possibility of encoding inputs. Surprisingly, this turned out not to be the case.

\paragraph{Our contribution: RE formalism $\approx$ general LCLs.}

An informal version of our main theorem is this:
\begin{quote}
    \em
    As long as we can break symmetry, the RE formalism is as expressive as general LCL problems.
\end{quote}
Here the ability to \emph{break symmetry} refers, in essence, to the ability to color a graph with a constant number of colors. For example, in the LOCAL and CONGEST models, this means that we will consider round complexities that are $\Omega(\log^* n)$. In many other models, such as finitely-dependent distributions, non-signaling distributions, SLOCAL, online-LOCAL, and dynamic-LOCAL, it is known that we can break symmetry essentially for free, and this assumption does not limit the applicability of our result at all.

\paragraph{Corollaries of our work.}

Essentially all counterintuitive properties of LCLs still hold for the RE formalism. Concrete corollaries of our work include:
\begin{itemize}
    \item There is an RE-formalism problem that admits a distributed quantum advantage in the LOCAL model, as a corollary of our work and \cite{balliu-casagrande-etal-2026-distributed-quantum}.
    \item There is an RE-formalism problem that benefits from shared randomness and from shared quantum state, as a corollary of our work and \cite{balliu-ghaffari-etal-2025-shared-randomness-helps-with}.
    \item There is an RE-formalism problem whose complexity in the LOCAL model depends on whether state transitions are restricted to computable functions, as a corollary of our work and \cite{local-computability-2026}.
    \item We can construct RE-formalism problems with round complexities such as $\Theta(\log^{123.45} n)$ and $\Theta(n^{0.12345})$, as a corollary of our work and \cite{balliu-hirvonen-etal-2018-new-classes-of-distributed}.
    \item Many questions related to the distributed complexity of RE-formalism problems are undecidable, as a corollary of our work and \cite{naor-stockmeyer-1995-what-can-be-computed-locally,brandt-hirvonen-etal-2017-lcl-problems-on-grids}.
\end{itemize}

\paragraph{Discussion.}

As the RE formalism is engineered to be \emph{over-restrictive}, it follows that many natural candidates for restricted LCLs are sandwiched between the RE formalism and general LCLs. For example, if we added input labels or generalized the formulation beyond $3$-regular graphs, we would have a formalism that is a strict generalization of the RE formalism. Hence all the counterintuitive properties discussed above still hold.

It is very hard to imagine a formalism of local graph problems that still contains problems similar in spirit to widely-studied graph problems, yet rules out the artificial problems discussed above.
Hence we arrive at these conclusions:
\begin{enumerate}
    \item \emph{The definition of LCL problems is extremely robust to variations.} This is great news, since this shows that the long line of prior work on LCL problems holds near-verbatim across a wide range of more restrictive problem families, all the way to input-free problems on regular graphs defined in the RE formalism. In essence, our community has studied \emph{the right problem family} all along; what we have discovered are fundamental truths about the nature of local problems, and not merely artifacts of the Naor--Stockmeyer definition of LCLs.
    \item \emph{The counterintuitive properties of LCLs are fundamental.} We cannot define away the artificial problems that we have seen. For example, the ability to encode grid-like constructions is also inherent in all restricted versions of LCLs. This is probably somewhat akin to the universality of the notion of computation, where we have seen that superficially-distinct models of computation turn out to be all equivalent to Turing machines.
\end{enumerate}

\paragraph{Future work and open questions.}

Our work leaves open two main questions:
\begin{enumerate}
    \item What, exactly, happens inside the ``symmetry-breaking region''? In particular, what can we say about the relative expressive power of general LCLs vs.\ the RE formalism when we zoom into $o(\log^* n)$-round LOCAL algorithms?
    \item What happens if we restrict not only the problem definition but also the family of graphs? In particular, what can we say about the case of trees?
\end{enumerate}
We also emphasize that, throughout this work, we will assume that our input graph is a simple graph. This does not make a difference in the case of LCL problems (for example, our LCL can trivially specify that nodes with multiple parallel edges or self-loops are unconstrained, and hence only make the problem easier to solve and will not appear in worst-case instances). However, our proof of the equivalence between LCLs and the RE formalism will make critical use of the assumption that the input graph in the RE formalism is indeed a simple graph. Whether allowing multiple parallel edges or self-loops there fundamentally changes the expressive power of the formalism is left as an open question for future work.

\subsection{Proof overview}

Our high-level plan is, in essence, this: we take any LCL problem $\Pi$, defined in the usual Naor--Stockmeyer formalism, and define a new graph problem $\Pi'$ that can be defined in the RE formalism. Then we show that, across all models of interest, $\Pi$ and $\Pi'$ have the same distributed computational complexity, as long as we are above the symmetry-breaking region.

The intuitive idea is easiest to explain, if we go through 5 different problem formulations, which we will call here $A$, $B$, $C$, $D$, and $E$. Here $A = \Pi$ is the original LCL problem, and $E = \Pi'$ is an essentially equivalent problem in the RE formalism. The formal proof will take a shortcut and skip $C$, so we will only see $A \to B \to D \to E$ later in this paper, but the $B \to D$ step in the formal proof captures the intuitive ideas that we explain below in steps $B \to C$ and $C \to D$.

\paragraph{\boldmath $A \to B$: making it regular and input-free.}

As a warm-up, let us first note that it is straightforward to turn an LCL $A$ into an equivalent LCL $B$ that is defined on input-free $3$-regular graphs. In essence, we can use simple local reductions to make the graph regular, and simple gadgets to encode inputs in the local structure of a regular graph. The hard part is to make $B$ verifiable while defining only the node and edge constraints (which cannot refer to the local graph structure).

\paragraph{\boldmath $B \to C$: making it PN-checkable.}

As an intermediate goal, let us consider the challenge of turning an LCL problem $B$ into an equivalent problem $C$ that is \emph{PN-checkable}. By a PN-checkable problem, we refer to a problem where the validity of the output can be verified with a distributed algorithm that works in the port-numbering model. In particular, the verifier cannot see short cycles; a verifier with locality $r$ can only observe the tree-like unfolding (local view) of each node up to distance $r$.

Now assume for a while that we \emph{can} encode inputs in problem $C$; the only challenge is that problem $B$ may refer to the local graph structure in general (say, ``the graph looks locally like a grid''), while $C$ can only specify valid input--output labelings for tree-like local views.

The trick that we can employ here is this: problem $C$ is defined so that, in addition to whatever labels we need for the original problem $B$, we also have an additional input label $x$ and an additional output label $y$ on each node. We then interpret $(x,y)$ as the locally unique identifier of a node. In particular, if we have two nodes in our local view with the same $(x,y)$ pair, we verify the output as if they were the same node, using the original constraints of $B$. We then need to argue that this results, indeed, in a problem that is equivalent to $B$; the intuition is as follows:
\begin{itemize}
    \item Assume that we have an algorithm $\mathcal{A}'$ for solving $C$. We can use it to construct an algorithm $\mathcal{A}$ for solving $B$ as follows. Algorithm $\mathcal{A}$ first finds a distance-$c$ coloring (for some suitable constant $c$) and puts it in the input labels $x$. Then we apply algorithm $\mathcal{A}'$. Note that regardless of what algorithm $\mathcal{A}'$ outputs in label $y$, we will have the following property: $(x,y)$ pairs will be distinct for any two distinct nodes in the same local view. Hence if $\mathcal{A}$ manages to make the PN-verifier of $C$ happy, the output of $\mathcal{A}$ will satisfy the constraints of the original problem $B$.
    \item Assume that we have an algorithm $\mathcal{A}$ for solving $B$. We can use it to construct an algorithm $\mathcal{A}'$ for solving $C$ as follows. For any input, $\mathcal{A}'$ simply ignores the label $x$ and applies the original algorithm $\mathcal{A}$. Then $\mathcal{A}'$ finds a distance-$c$ coloring and puts it in the output labels $y$. This will again ensure that $(x,y)$ pairs are distinct for any two distinct nodes in the same local view. Hence if $\mathcal{A}'$ manages to satisfy the constraints of the original problem $B$, the output of $\mathcal{A}'$ will make the PN-verifier of $C$ happy.
\end{itemize}
Note that here we will need to break symmetry and find a distance-$c$ coloring. Hence, for example, when applied to the LOCAL model, this means that there is an additive $O(\log^* n)$ overhead. Put differently, the round complexities of $B$ and $C$ are asymptotically the same as long as we are in the $\Omega(\log^* n)$ region.

\paragraph{\boldmath $C \to D$: eliminating inputs again.}

So far we have been able to eliminate the inputs and make the graph regular. Then we made the problem PN-checkable, but to do that we needed to introduce inputs. Now we will need to eliminate such inputs again, yet keep the problem PN-checkable.

This is the most challenging technical obstacle that we need to overcome in this work. Intuitively, the idea is that we replace the input labels with a sequence of gadgets. However, it is not at all obvious that there is any such gadget that we can use so that our original problem $C$ and the new problem $D$ have the same complexity. The challenges are as follows:
\begin{itemize}
    \item An algorithm $\mathcal{A}'$ that solves $D$ cannot assume anything about the input. The input may be adversarial. It may not consist of an appropriate sequence of gadgets. Yet $\mathcal{A}'$ must be able to produce a feasible solution for any input. In essence, we would like to devise a way for $\mathcal{A}'$ to announce that the input structure is invalid and hence there is no need to do anything.
    \item However, if $\mathcal{A}'$ can announce that the input is invalid, it may open up a loophole that makes problem $D$ trivial. Perhaps $\mathcal{A}'$ can \emph{always} claim that the input is invalid? The PN-verifier cannot ``see'' the structure of the gadgets, so it cannot directly know whether $\mathcal{A}'$ is lying when it claims that the input is invalid.
\end{itemize}

We overcome this challenge in \cref{sec:graph-encoding} by introducing a very specific gadget structure and requiring the algorithm to produce an output labeling $\chi$ that looks like a consistent distance-$c$ coloring of the graph in all local neighborhoods. The key graph-theoretic result that we need to prove is, informally, this (see \cref{thm:gadgeting}):
\begin{quote}
    If the input follows the right gadget structure, then the only locally-consistent labeling $\chi$ has to be a distance-$c$ coloring of the graph.
\end{quote}
Notably, an algorithm cannot cheat and label two nearby nodes $u$ and $v$ with the same label and in this way trick the verifier into believing that they are indeed the same node in the local view.

Now problem $D$ is defined so that the labeling $\chi$ must be locally-consistent. Then the PN-verifier can ``decode'' the gadgets and recover the underlying input for our original LCL problem $C$ and apply the verifier of $C$. Again, we need to argue that this results in a problem that is equivalent to $C$; the intuition is as follows:
\begin{itemize}
    \item Assume that we have an algorithm $\mathcal{A}'$ for solving $D$. We can use it to construct an algorithm $\mathcal{A}$ for solving $C$ as follows. Algorithm $\mathcal{A}$ first breaks symmetry and constructs appropriate gadget sequences to form an input for $D$, and applies algorithm $\mathcal{A}'$. Algorithm $\mathcal{A}'$ does not have any room for cheating, and it has to produce a genuine distance-$c$ coloring $\chi$, together with an output that makes the verifier of $D$ happy; hence the output (with appropriate translations) must satisfy the constraints of the original problem $C$.
    \item Assume that we have an algorithm $\mathcal{A}$ for solving $C$. We can use it to construct an algorithm $\mathcal{A}'$ for solving $D$ as follows. For any input, $\mathcal{A}'$ first finds a genuine distance-$c$ coloring $\chi$ of the graph. Then $\mathcal{A}'$ tries to decode the gadget structure. If successful, it will apply $\mathcal{A}$ locally. If unsuccessful, it can indicate that the gadget structure is invalid, and $\chi$ will serve as a PN-checkable proof of this fact. In both cases the output of $\mathcal{A}'$ will make the PN-verifier of $D$ happy.
\end{itemize}

\paragraph{\boldmath $D \to E$: from PN-checkable to the RE formalism.}

Finally, we need to turn PN-checkable problems into problems in the RE formalism. In essence, the RE formalism is a restrictive special case of PN-checkable problems with a very small verification radius. We apply the usual technique of reducing the verification radius: each edge has to produce an output label that encodes the radius-$r$ local views of its endpoints.

\paragraph{Going back from the RE formalism to LCL problems.}

So far we have seen that anything that can be expressed as an LCL problem can be also expressed as an equivalent problem in the RE formalism (modulo some overhead related to symmetry-breaking). Conversely, anything that can be expressed in the RE formalism can be trivially written as an equivalent LCL problem. This completes the proof showing that the two formalisms (as well as anything in between them) are equally expressive.

\subsection{Structure of this work}

\begin{itemize}
    \item \cref{sec:prelim}: Definitions of the relevant restricted variants of LCL problems (\cref{ssec:lcls}) and the relevant models of computing (\cref{ssec:models}).
    \item \cref{sec:framework}: General framework that allows us to do ``local reductions'' between problems and formalisms, across different models of computing.
    \item \cref{sec:eliminate-inputs}: Eliminating inputs and making the graph regular.
    \item \cref{sec:graph-encoding}: Making the problem PN-checkable without inputs. \emph{\textbf{This is the interesting and nontrivial part of the paper, where heavy lifting takes place.}}
    \item \cref{sec:re-formalism}: Final steps from PN-checkable problems to the RE formalism.
    \item \cref{sec:final}: Putting the pieces together and formally stating the main theorem.
\end{itemize}

\section{Preliminaries}\label{sec:prelim}

We write $\mathbb{N} = \{0,1,2,\dotsc\}$, $\mathbb{N}_+ = \{1,2,\dotsc\}$, $[n] = \{1, 2, \dotsc, n\}$, and $[0] = \emptyset$.

\subsection{Graphs}

\paragraph{Basic concepts.}
In this paper, graphs are assumed to be simple and undirected, unless specified otherwise. A graph $G = (V, E)$ consists of a set of nodes (or vertices) $V$ and a set of edges $E$. When necessary, we write $V(G)$ and $E(G)$ to explicitly refer to the vertex and edge sets of a specific graph $G$.
The \emph{distance} between two nodes $u, v \in V(G)$, denoted by $\dist_G(u, v)$, is defined as the number of edges on a shortest path between $u$ and $v$. The \emph{diameter} of a graph $G$, defined as $\diam(G) = \max_{u,v\in V(G)}\dist_G(u,v)$, is the length of the longest shortest path in $G$. When the graph is clear from the context, we omit the subscript and simply write $\dist(u,v)$.
The \emph{degree} of a node $v \in V(G)$, denoted $\deg_G(v)$, is the number of edges incident to $v$. We write $\Delta(G) = \Delta = \max_{u\in V(G)}\deg_G(u)$ for the maximum degree of $G$.

\paragraph{Neighborhoods, balls, and centered graphs.}
For a vertex $v$, we write \[\mathcal{N}(v) = \{ u \in V : \{u,v\} \in E \}\] for its \emph{neighborhood}, i.e., the set of adjacent nodes. We define the \emph{ball of radius $r$ around $v$} as $B_G^r(v) = (V', E')$, where
\begin{align*}
V' &= \{ u \in V : \dist(v,u)\leq r \}, \\
E' &= \{ \{s,t\} \in E : \dist(v,s) < r \text{ or } \dist(v,t) < r \}.
\end{align*}
Finally, a pair $(H,v)$ is called a \emph{centered graph of radius $r$} if $B_H^r(v) \cong H$, that is, $H$ is fully contained in the ball of radius $r$ around $v$. We extend these notions in natural ways to node-labeled and edge-labeled graphs.

\paragraph{PN-views.}

A \emph{walk} of length $\ell$ is a sequence of nodes $(v_0, \dotsc, v_\ell)$ such that $\{v_i, v_{i+1}\} \in E$ for all $i$. A walk is \emph{non-backtracking} if $v_i \ne v_{i+2}$ for all $i$. Let $V_G^r(v)$ consist of all non-backtracking walks of length at most $r$ that start at $v_0 = v$. Define the successor relation $A_G^r(v)$ as follows: walk $(v_0, \dotsc, v_\ell) \in V_G^r(v)$ is a successor of walk $(v_0, \dotsc, v_{\ell-1})$. Now define a tree $T^r_G(v)$ with the set of nodes $V_G^r(v)$ and the set of directed edges $A_G^r(v)$. Finally, if $G$ is a labeled graph, the label of walk $(v_0, \dotsc, v_\ell)$ in $T^r_G(v)$ is inherited from the label of node $v_\ell$ in $G$; we define the edge labels in an analogous manner.
Now we say that any labeled tree isomorphic to $T^r_G(v)$ is the \emph{radius-$r$ PN-view} of node $v$.

Importantly, the local view does \emph{not} carry information on the identity of original nodes; it only carries information on the structure of the tree formed by non-backtracking walks. For example, the local views in a $3$-cycle and $10$-cycle are isomorphic---for any $r$ they are simply a node and two directed paths of length $r$. Intuitively, a distributed algorithm that works in an anonymous network can recover the structure of $T^r_G(v)$ but it cannot recover the structure of $B^r_G(v)$.

Finally, a \emph{centered PN-view of radius $r$} is a pair $(T,v)$ such that $T$ is a tree isomorphic to the radius-$r$ local view of $v$ in $T$.

\begin{remark}
Here the name ``PN'' refers to the \emph{port-numbering} model, perhaps the most commonly-used formalization of distributed algorithms in anonymous networks. However, we do not need to worry about the concept of port-numbered networks and port numbers in this work.
\end{remark}


\subsection{LCL problems and their restrictions}\label{ssec:lcls}

In general, a graph problem $\Pi$ associates a set of input-output-labeled graphs $(G, \sigma^{\inpt}, \sigma^{\oupt})$ with each input-labeled graph $(G, \sigma^{\inpt})$; we will call these \emph{valid solutions}, \emph{legal solutions}, or simply \emph{solutions}. We are given $(G, \sigma^{\inpt})$ as input, and the algorithm that runs in $G$ has to produce an output labeling $\sigma^{\oupt}$ such that $(G, \sigma^{\inpt}, \sigma^{\oupt})$ is a valid solution.

\paragraph{LCL problems.}

A \emph{locally checkable labeling} problem \cite{naor-stockmeyer-1995-what-can-be-computed-locally}, or in brief an LCL, is defined by a finite set $\CC$ of input-output-labeled centered graphs of radius $r$.
We define that a solution $\mathbf{G} = (G, \sigma^{\inpt}, \sigma^{\oupt})$ is valid if $(B_{\mathbf{G}}^r(v), v)$ is isomorphic to a centered graph in $\CC$ for all nodes $v$.

Intuitively, these are problems where the validity of a solution can be verified with an $r$-round LOCAL algorithm, in the following sense: all nodes output ``yes'' if a solution is valid and at least one node outputs ``no'' if a solution is not valid. To verify the solution, it suffices that all nodes gather their radius-$r$ labeled neighborhood and check that it is indeed isomorphic to a centered graph in $\CC$.

Since $\CC$ is finite, it also implies that the problem is well-defined only for graphs of some finite maximum degree $\Delta$. Hence, throughout this paper, we will focus on the family of graphs of maximum degree $\Delta$, for some constant $\Delta$.
Furthermore, since $\CC$ is finite, there are only finitely many possible input labels and output labels that we can use. We will write $\inlabels$ and $\outlabels$ for the sets of input and output labels. An LCL problem with all relevant sets and parameters can then be fully specified as a tuple \[\Pi=\left(\Delta,\inlabels,\outlabels,\CC,r\right).\]

\paragraph{PN-checkable problems.}

Intuitively, an LCL $\Pi$ is \emph{PN-checkable} if the validity of a solution can be verified with an $r$-round algorithm in the port-numbering model. Such algorithms do not have access to unique identifiers and hence cannot e.g.\ distinguish between a short cycle and an infinitely long path.

More formally, a PN-checkable problem is defined by a finite set $\CC$ of centered PN-views of radius $r$.
We define that a solution $\mathbf{G} = (G, \sigma^{\inpt}, \sigma^{\oupt})$ is valid if $(T_{\mathbf{G}}^r(v), v)$ is isomorphic to a centered PN-view in $\CC$ for all nodes $v$.

Importantly, LCL problems can directly capture tasks such as ``a node must be labeled with $1$ if and only if it is part of a $3$-cycle'', while PN-checkable problems cannot directly refer to the existence of short cycles.

\paragraph{RE formalism.}

Intuitively, a problem in the \emph{RE formalism} is a very restricted PN-checkable problem, with the checking radius ``$r=1/2$,'' with labels on \emph{half-edges}, and we are only interested in the case of $3$-regular graphs. The formalism is essentially a restricted version of the notation used to specify problems in the Round Eliminator tool \cite{olivetti-2020-brief-announcement-round-eliminator-a,olivetti-2025-round-eliminator-a-tool-for-automatic}.

A problem in the RE formalism is a tuple
\[
\Pi = (\outlabels, \CC_V, \CC_E),
\]
where $\outlabels$ is a finite set of edge labels, $\CC_V$ is the \emph{node constraint} and $\CC_E$ is the \emph{edge constraint}. Here $\CC_V$ is a set of $3$-element multisets from $\outlabels$, and $\CC_E$ is a set of $2$-element multisets from $\outlabels$.

In the RE formalism, we will completely ignore the input labeling $\sigma^{\inpt}$ (or, equivalently, require that the input labels are empty). Our graph $G$ is assumed to be $3$-regular (almost equivalently, we could specify that nodes of degree different from $3$ are unconstrained, but let us stick to the $3$-regular case for simplicity). Our output labeling $\sigma^{\oupt}$ assigns a label from $\outlabels$ to each half-edge (put otherwise, we can imagine that we have subdivided all edges of $G$ to construct a new graph $G'$ with twice as many edges, and $\sigma^{\oupt}$ is an edge-labeling of $G'$).

Now a half-edge labeling $\sigma^{\oupt}$ is valid if (1)~for every node $v$, the three half-edge labels incident to it form a multiset that is in $\CC_V$, and (2)~for every edge $e$, the two half-edge labels on its two ends form a multiset that is in $\CC_E$.

\paragraph{Always-solvable LCLs.}
We say that an LCL $\Pi$ is \emph{always solvable} if for every graph $G$ of maximum degree $\Delta$ and every input labeling $\sigma^{\inpt}$, there exists at least one legal output labeling $\sigma^{\oupt}$.


\subsection{Key models of computing}\label{ssec:models}

We will here briefly introduce the key models of computing that we will discuss in this work. In essence, we want to show that our main result---that problems in the RE formalism are equally expressive as general LCL problems---holds across a wide range of models of computing, and hence we will need to refer to various models introduced and studied in prior work. The unifying theme here is that all models of interest are parameterized by $T$, which represents the \emph{locality} or \emph{round complexity} of the algorithm.

\paragraph{The \local model.}
In the \local model \cite{linial-1992-locality-in-distributed-graph-algorithms,peleg-2000-distributed-computing-a-locality-sensitive} of computing, we are given a distributed system of $n$ processors (or nodes) connected through a communication network represented as a graph $G=(V,E)$, along with an input function $x$. Every node $v\in V(G)$ holds input data $x(v)$ which encodes the number $n$ of nodes in the network, a unique identifier from the set $[n^c]$, where $c\geq 1$ is a fixed constant, and \emph{possible} inputs defined by the problem of interest. In case of randomized computation we refer to the \randlocal (i.e., randomized \local), which means that $x(v)$ additionally encodes an \emph{infinite} string of bits that are uniformly and independently sampled for each node, and private (i.e., not shared with the other nodes in the graph). In case of \emph{deterministic} computation, we call the model \detlocal. Computation is performed by synchronous rounds of communication. In each round, nodes can exchange messages of \emph{unbounded} (but finite) size with their neighbors, and then perform an arbitrarily long (but terminating) local computation. Errors occur neither in sending messages nor during local computation. The computation terminates when every node $v$ outputs a label $\sigma$. The running time of an algorithm in this model is the number of communication rounds, given as a function of $n$, that are needed to output a labeling that solves the problem of interest. In \randlocal, we also require that the algorithm solves the problem of interest with probability at least $1-\delta$ where $\delta=1/\poly(n)$, where $\poly(n)$ is any polynomial function in $n$. If an algorithm runs in $T$ rounds, and both communication and computation are unbounded, we can look at it as a function mapping the set of radius-$T$ neighborhoods to the set of output labels in the deterministic case, or to a probability distribution over the output labels in the randomized case. In this case, we say that $T$ is the \emph{locality} of the algorithm.
\paragraph{The \slocal model.}
The \slocal model of computing~\cite{ghaffari-kuhn-maus-2017-on-the-complexity-of-local} is a sequential version of the \local model. More precisely, an algorithm $\AA$ processes the nodes sequentially in an order $p=v_1,v_2,\dots,v_n$. The algorithm must work for any given order $p$. When processing a node $v$, the algorithm can query $B_G^T(v)$ and $\AA$ can read $u$'s state for all nodes $u\in B_G^T(v)$. Based on this information, node $v$ updates its own state and computes its output label $\sigma$. In doing so, $v$ can perform unbounded computation, i.e., $v$'s new state can be an arbitrary function of the queried $B_G^T(v)$. The output can be remembered as a part of $v$'s state. The \emph{time complexity} $T_{\AA,p}(G,\bm{x})$ of the algorithm on graph $G$ and inputs $\bm{x}=(x(v_1),\dots,x(v_n))$ with respect to order $p$ is defined as the maximum $T$ over all nodes $v$ for which the algorithm queries a radius-$T$ neighborhood of $v$. The \emph{time complexity} $T_\AA$ of algorithm $\AA$ on graph $G$ and inputs $\bm{x}$ is the maximum $T_{\AA,p}(G,\bm{x})$ over all orders $p$.


\paragraph{The \nonsign model.} The \nonsign model is a model of computing that abstracts from how the actual computation is happening in the network, focusing on a probabilistic description of the valid output labelings. In this model, rather than given algorithms, we are asked to produce \emph{outcomes} (also called strategies) that are functions mapping the input to a probability distribution over output labelings. In other words, given a graph and its input, a probability distribution over output labelings is assigned to the graph such that a sample from the distribution will produce a valid output labeling with high probability. The complexity of the outcome is given by its dependency radius \(T\). This means that an outcome is non-signaling beyond distance \(T\) if, for any subset \(A\) of the nodes of the graph, modifying the graph or its input at distance greater than \(T\) from \(A\) does not change the output distribution over \(A\).
Different papers refer to this model as \emph{non-signaling distributions}, \emph{NS-LOCAL}, \emph{physical locality}, and the \emph{$\phi$-LOCAL} model; we refer to \cite{gavoille-kosowski-markiewicz-2009-what-can-be-observed,arfaoui-fraigniaud-2014-what-can-be-computed-without,akbari-coiteux-roy-etal-2025-online-locality-meets} for more details.

\paragraph{The bounded-dependence model.}
In the bounded-dependence model, we have a probability distribution that is not only non-signaling, but it also satisfies the following independence property: if $A$ and $B$ are two sets of nodes such that their radius-$T$ balls do not intersect, then outputs of $A$ and $B$ are independent. In the case of a constant $T$, this coincides (up to constant factors) with the notion of \emph{finitely dependent distributions} \cite{burton-goulet-meester-1993-on-1-dependent-processes-and,holroyd-liggett-2016-finitely-dependent-coloring}. We refer to \cite{akbari-coiteux-roy-etal-2025-online-locality-meets} for more details.

\paragraph{The \qlocal model.} The \qlocal model is defined like the \local model introduced above, with the following differences. Every processor (node) can locally operate on an unbounded (but finite) number of qubits, applying any unitary transformations, and quantum measurements can be locally performed by nodes at any time. In each communication round, nodes can send an unbounded (but finite) number of qubits to their neighbors. The local output of a node still needs to be an output label encoded in classical bits. As in \randlocal, we ask that an algorithm solves a problem with probability at least $1-1/\poly(n)$. A more formal definition of the model can be found in~\cite{gavoille-kosowski-markiewicz-2009-what-can-be-observed}.

\section{General framework}\label{sec:framework}
This section introduces the generic gadget-based framework that we use as a template for
reductions between LCL problems and formalisms across multiple models of computing.
We first define an encoding that maps any bounded-degree labeled instance to a $3$-regular gadgeted
graph while preserving the underlying structure up to isomorphism (\cref{sec:graph_encoding}). We then specify a
decoding procedure together with the associated
$\decode{\cdot }{\cdot}$ correspondence between node objects and gadget vertices (\cref{sec:graph_decoding}).
Building on this, we define natural lifts between output labelings of the decoded instance and
the gadget graph, and show how to simulate algorithms through
the encoding/decoding with only constant-factor overhead determined by the maximum gadget
diameter $\Lambda$ (\cref{sec:simulation_of_algorithms}). Finally, we describe the symmetry-breaking complexity for the models of interest as an additive model-dependent cost $\diamondsuit(\mathcal{M})$ (\cref{sec:symmetry_breaking}).
\subsection{Graph encoding}\label{sec:graph_encoding}

We describe a general  graph encoding scheme that transforms a bounded-degree labeled graph into a $3$-regular graph.

Let $G_1=(V_1,E_1)$ be a graph of maximum degree at most $d$. Each vertex $v\in V_1$ is assigned a label $\ssigma{V_1}{}(v)\in [k_1]$, and each edge $e\in E_1$ is assigned a label $\ssigma{E_1}{}(e)\in [k_2]$, where $k_1,k_2\ge 0$ are fixed integers.

For every vertex $v\in V_1$, we define a \emph{node gadget} $\ngadget{v}{\ssigma{V_1}{}(v)}$ that encodes the pair $(v,\ssigma{V_1}{}(v))$. In the encoded graph, each vertex of $G_1$ is replaced by its corresponding node gadget.

Similarly, for every edge $e\in E_1$, we define an \emph{edge gadget} $\egadget{e}{\ssigma{E_1}{}(e)}$ that encodes the pair $(e,\ssigma{E_1}{}(e))$. Each edge of $G_1$ is replaced by the corresponding edge gadget, which connects the node gadgets of its endpoints.

The structure of node and edge gadgets is completely determined by the labeling functions $\ssigma{V_1}{}$ and $\ssigma{E_1}{}$, together with the degree of the corresponding vertices. Since the label sets are finite and the degree of $G_1$ is bounded, there exist only finitely many distinct node and edge gadgets, each of finite size.

The resulting encoded graph $G'_1$ is always $3$-regular. The encoding scheme  must satisfy the following properties.

\begin{description}
	\item[Uniqueness of encoding.] Let there be another labeled graph $(G_2,\ssigma{V(G_2)}{},\ssigma{E(G_2)}{})$ which is encoded into a $3$-regular graph $G'_2$. If $G'_1$ and $G'_2$ are isomorphic, then $(G_1,\ssigma{V(G_1)}{},\ssigma{E(G_1)}{})$ and $(G_2,\ssigma{V(G_2)}{},\ssigma{E(G_2)}{})$  are also isomorphic.
	\item[Decodability assumptions.] We assume that the encoding scheme satisfies that in the graph $G'_1$ produced by the encoding, each vertex of $G'_1$ belongs to at most one gadget, node and edge gadgets do not overlap, and every detected gadget corresponds to a unique vertex or edge of the encoded graph.
\end{description}

\subsection{Graph decoding}\label{sec:graph_decoding}

We define a \emph{graph decoding} from the family of $3$-regular graphs to the family of finitely labeled graphs of maximum degree at most $d$.

 Let $G'$ be a $3$-regular graph  and its decoding is $(G,\ssigma{V(G)}{},\ssigma{E(G)}{})$ where $\ssigma{V(G)}{} :V(G)\rightarrow [k_1]$ and $\ssigma{E(G)}{} : E(G)\rightarrow [k_2]$ be the input labelings of $G$ where $k_1$ and $k_2$ are fixed positive integers. 	The decoding procedure also implicitly defines a function \[\mathsf{decode}^{-1}:(V(G)\times [k_1])\cup (E(G)\times [k_2])\rightarrow 2^{V(G')},\] which is specified alongside the decoding procedure below. We fix arbitrary values $\bot_V \in [k_1]$ and $\bot_E \in [k_2]$, which are used as default labels. For simplicity of notation, we denote both values by $\bot$. The graph decoding proceeds in three steps.

\noindent
\paragraph{Detect node gadgets.}
 
We identify induced subgraphs of $G'$ which are isomorphic to one of the possible node gadgets. For each detected node gadget which encodes the integer $k\in [k_1]$, we create a node $v\in V(G)$ and set $\ssigma{V(G)}{}(v) = k$. We refer to the node gadget as $\ngadget{v}{k}$. We set  $\decode{v}{k} = \ngadget{v}{k}$. 

\noindent
\paragraph{Detect edge gadgets.}
If there exists an edge gadget encoding an integer $p\in [k_2]$ connecting two node gadgets  $\ngadget{u_1}{k'_1}$ and $\ngadget{u_2}{k'_2}$,  then we add the edge $\{u_1,u_2\}\in E(G)$ and ensure that $\ssigma{E(G)}{}(\{u_1,u_2\}) = p$. We set $\decode{\{u_1,u_2\}}{p} = \egadget{\{u_1,u_2\}}{p}$. 

\noindent
 \paragraph{Malformed parts.}
Any vertex or edge of $G'$ that is not contained in any detected node gadget or detected edge gadget is considered \emph{malformed}.

		Every malformed vertex $v$ of $G'$ is lifted to the decoded graph $G$ as a vertex and it takes the  default value $\bot$ in $[k_1]$. We set $\decode{v}{\bot} = \{v\}$. 

		Each malformed edge $\{u_1,u_2\}$ of $G'$  gives rise to an edge in the decoded graph (adopting a default label) according to its endpoints:
    \begin{itemize}
		    \item If both $u_1$ and $u_2$  are malformed vertices, then the edge $\{u_1,u_2\}$ is lifted as is in $G$ and $\ssigma{E(G)}{}(\{u_1,u_2\}) = \bot$. Set $\decode{\{u_1,u_2\}}{\bot} = \{u_1,u_2\}$.
		    \item If $u_1$ belongs to a detected node gadget corresponding to a vertex $v\in V(G)$ and $u_2$ is malformed, then $\{v,u_2\}$ is added in $G$ and $\ssigma{E(G)}{}(\{v,u_2\}) = \bot$. Set $\decode{\{v,u_2\}}{\bot} = \{u_1,u_2\}$ 
		    \item If $u_1$ and $u_2$  belong to two detected node gadgets corresponding to vertices $u,v \in V(G)$, then an edge is added between $u$ and $v$ in $G$ and $\ssigma{E(G)}{}(\{u,v\}) = \bot$. Set $\decode{\{u,v\}}{\bot} = \{u_1,u_2\}$.  
    \end{itemize}

The decoded graph $G$ is allowed to be a multigraph; in particular, multiple edges between the same pair of vertices may arise from malformed edges.

\begin{remark}
Even though $\decode{\cdot}{\cdot}$ outputs a set of vertices, we abuse the notation and assume that it returns a graph.
\end{remark}

\begin{remark}
The function $\decode{\cdot}{\cdot}$ also implicitly defines $\mathsf{decode}$ which is just its preimage.
\end{remark}

\subsubsection{LOCAL decoding}

Fix $\Lambda$ to be thrice the maximum diameter of any node gadget or edge gadget.  The following lemma defines the local decoding. 

\begin{lemma}\label{lemma:local_decoding}
	Fix an arbitrary node $v\in V(G')$ such that it is not a part of any edge gadget. Let $(G'',\ssigma{V(G'')}{},\ssigma{E(G'')}{})$ be the decoding of  $\ball{\Lambda}{G'}{v}$. Let $u\in V(G)$ and $u'\in V(G'')$ be the decoded nodes such that $v$ belongs to both $\decode{u}{\ssigma{G}{}(u)}$ and $\decode{u'}{\ssigma{G''}{}(u')}$. 
Then
	\[
		\bigl(\ball{}{G}{u},\rsigma{V(G)}{}{\ball{}{G}{u}}, \rsigma{E(G)}{}{\ball{}{G}{u}}\bigr)
       				 \cong
        		       \bigl(\ball{}{G''}{u'},\rsigma{V(G'')}{}{\ball{}{G''}{u'}},\rsigma{E(G'')}{}{\ball{}{G''}{u'}}\bigr).
	\]
\end{lemma}
\begin{proof}
	In order to prove an isomorphism, we describe a mapping $\varphi: V(\ball{}{G}{u})\rightarrow V(\ball{}{G''}{u'})$ as follows. 
   We first set $\varphi(u) = u'$ and $\ssigma{G''}{}(u') = \ssigma{G}{}(u)$. The edge $\{u,w\}\in E(\ball{}{G}{u})$ if and only if $\decode{u}{\ssigma{V(G)}{}(u)}$ and $\decode{w}{\ssigma{V(G)}{}(w)}$ is connected with $\decode{\{u,w\}}{\ssigma{E(G)}{}(\{u,w\})}$. Since $\decode{u}{\ssigma{V(G)}{}(u)}$,  $\decode{w}{\ssigma{V(G)}{}(w)}$ and $\decode{\{u,w\}}{\ssigma{E(G)}{}(\{u,w\})}$ are contained completely within $\ball{\Lambda}{G'}{v}$ (as the maximum diameter of a node or an edge gadget is at most $\Lambda/3$ and the malformed nodes and edges have diameter at most $1$), the decoding procedure of  $\ball{\Lambda}{G'}{v}$ creates a neighbor $\varphi(w)\in \ball{}{G''}{u'}$ and sets $\ssigma{V(G'')}{}(\varphi(w))= \ssigma{V(G)}{}(w)$ and $\ssigma{E(G'')}{}{\{u',\varphi(w)\}} = \ssigma{E(G)}{}(\{u,w\})$  thus describing the bijective mapping.
\end{proof}


\begin{corollary}\label{cor:local_decoding_radius_T}
For any integer $T > 0$, fix an arbitrary node $v\in V(G')$ such that it is not a part of any edge gadget. Let $(G'',\ssigma{V(G'')}{},\ssigma{E(G'')}{})$ be the decoding of  $\ball{\Lambda T}{G'}{v}$. Let $u\in V(G)$ and $u'\in V(G'')$ be the decoded nodes  such that $v$ belongs to both $\decode{u}{\ssigma{G}{}(u)}$ and $\decode{u'}{\ssigma{G''}{}(u')}$.
Then
\[
\bigl( \ball{T}{G}{u}, \rsigma{V(G)}{}{\ball{T}{G}{u}}, \rsigma{E(G)}{}{\ball{T}{G}{u}} \bigr)
\cong
\bigl( \ball{T}{G''}{u'}, \rsigma{V(G'')}{}{\ball{T}{G''}{u'}}, \rsigma{E(G'')}{}{\ball{T}{G''}{u'}} \bigr).
\]
\end{corollary}

\begin{proof}
		We prove the claim by induction. For base case, $T=1$ and it follows immediately from \cref{lemma:local_decoding}. Let the claim be true for $T=i-1$ for some integer $i>1$. Let the isomorphism between $\ball{i-1}{G}{u}$ and  $\ball{i-1}{G''}{u'}$ preserving the node and edge labeling  be $\varphi$. Let $w\in V(G)$ such that $\dist_{G}(u,w) = i-1$. By inductive hypothesis, $\decode{w}{\ssigma{V(G)}{}(w)}$ is contained within $\ball{\Lambda(i-1)}{G'}{v}$.  For every node $q\in \ball{}{G}{w}$, $\decode{q}{\ssigma{V(G)}{}(q)}$ and $\decode{\{w,q\}}{\ssigma{E(G)}{}(\{w,q\})}$  are contained completely within $\ball{\Lambda i}{G'}{v}$ due to the same reason as in \cref{lemma:local_decoding} and therefore, the decoding procedure creates a neighbor $\varphi(q)$ of $\varphi(w)$ in $G''$. This defines an extension of $\varphi$ as an isomorphism between $\ball{i}{G}{u}$ and $\ball{i}{G''}{u'}$  proving the hypothesis.
\end{proof}

\begin{lemma}[\local Decoding] \label{lemma:local_decoding_procedure}
	There exists a deterministic \local algorithm with locality $O(\Lambda)$ that, running on a node $v$ in  $G'$, decides whether it is a part of $\decode{u}{\ssigma{V(G)}{}(u)}$ for some $u\in V(G)$. If it is, then it recognizes a unique representative node  which is consistently recognized by every node in $\decode{u}{\ssigma{V(G)}{}(u)}$. If $v$ recognizes itself, then it outputs the following:
	\begin{enumerate}
		\item Output the tuple $(u,\ball{}{G}{u},\rsigma{V(G)}{}{\ball{}{G}{u}},\rsigma{E(G)}{}{\ball{}{G}{u}})$.\label{p1}
		\item Assign a unique identifier to $u$ and  every other node $\ball{}{G}{u}$ represented as a function  $\id'$ whose co-domain  is bounded by $\poly(|V(G')|)$.\label{p2}
		\item For each node $w\in V(\ball{}{G}{u})$, node $v$ stores a path $P_{w:v,x_w}$ between nodes $v$ and $x_w$ in $G'$ where  $x_w$ is the representative node of $\decode{w}{\ssigma{V(G)}{}(w)}$.\label{p3} 
	\end{enumerate}
\end{lemma}

\begin{proof}
	We describe a deterministic LOCAL procedure executed at node $v$ as follows.  
	\begin{enumerate}
		\item Node $v$ gathers its radius-$\Lambda$ neighborhood $\ball{\Lambda}{G'}{v}$. Since every node gadget and edge gadget has diameter at most $\Lambda/3$, this view contains the entire gadget (if any) that $v$ belongs to and hence, decides whether it is a part of $\decode{u}{\ssigma{V(G)}{}(u)}$ for some node $u\in V(G)$.

		\item Let us assume that the node  $v$ belongs to  $\decode{u}{\ssigma{V(G)}{}(u)}$. The interesting case is when $\decode{u}{\ssigma{V(G)}{}(u)}$ is the node gadget $\ngadget{u}{\ssigma{V(G)}{}(u)}$. Then $v$ recognizes the node with the minimum $\id$ in the node gadget as the representative node. This choice is consistent across all nodes of $\decode{u}{\ssigma{V(G)}{}(u)}$. It does the same with $\decode{w}{\ssigma{V(G)}{}(w)}$ for each $w\in \ball{}{G}{u}$.

		\item Let us assume that the node $v$ has the minimum $\id$ in $\ngadget{u}{\ssigma{V(G)}{}(u)}$. It  then applies the graph decoding procedure on  $\ball{\Lambda}{G'}{v}$. Let the decoded graph-labeling pair be $(G'',\ssigma{V(G'')}{},\ssigma{E(G'')}{})$ and let  $u'\in V(G'')$ be a node such that $v\in \ngadget{u'}{\ssigma{V(G'')}{}(u')}$. The node $v$ then outputs  $(u,\ball{}{G''}{u'},\rsigma{V(G'')}{}{\ball{}{G''}{u'}}, \rsigma{E(G'')}{}{\ball{}{G''}{u'}})$ and sets $\id'(u')\coloneqq\id(v)$. 
		
		\item For each $w'\in V(\ball{}{G''}{u'})$, let $x_{w'}$ be the representative node of $\decode{w'}{\ssigma{V(G'')}{}(w')}$. Set $P_{w':v,x_{w'}}$ as the shortest path between $v$ and $x_{w'}$ in $G'$.  
	\end{enumerate}
	By \cref{lemma:local_decoding}, we have that \[(\ball{}{G}{u},\rsigma{V(G)}{}{\ball{}{G}{u}}, \rsigma{E(G)}{}{\ball{}{G}{u}})\cong   ( \ball{}{G''}{u'},\rsigma{V(G'')}{}{\ball{}{G''}{u'}}, \rsigma{E(G'')}{}{\ball{}{G''}{u'}}).\] By the isomorphism and the description of the algorithm, items~\ref{p1},~\ref{p2} and~\ref{p3}  of the lemma follow immediately.
\end{proof}

\subsubsection{LOCAL encoding}

As in local decoding, we present lemmas which encapsulate the local version of graph encoding. As specified in \cref{sec:graph_encoding}, we begin with $(G_1,\ssigma{V(G_1)}{},\ssigma{E(G_1)}{})$ where  $G_1$ has max. degree at most $d$ and its encoding $G'_1$ which is $3$-regular.


\begin{lemma}[Locality of encoding]\label{lem:locality_of_encoding}
For any integer $T\ge 0$ and any vertex $x\in V(G_1)$, let $H'$ be the encoded graph of $(\ball{T+1}{G_1}{x},\rsigma{V_1}{}{\ball{T+1}{G_1}{x}},\rsigma{E_1}{}{\ball{T+1}{G_1}{x}})$.
Let $a$ be an arbitrary vertex of the node gadget $\ngadget{x}{\ssigma{V_1}{}(x)}$ in $G'_1$, and let $a_H$ denote the corresponding vertex of the same node gadget 
in $H'$. Then $\ball{T}{G'_1}{a}\ \cong\ \ball{T}{H'}{a_H}$.

\end{lemma}

\begin{proof}
Let $S\subseteq V(G'_1)$ be the set of vertices that belong to gadgets corresponding to vertices or edges of $\ball{T+1}{G_1}{x}$. By the encoding assumptions from \cref{sec:graph_encoding} (gadgets are disjoint, and each vertex/edge of $G_1$ gives rise to a unique gadget whose internal structure is determined by its label and degree), for every $y\in V(\ball{T+1}{G_1}{x})$ the node gadget $\ngadget{y}{\ssigma{V_1}{}(y)}$ appears in $G'_1$ and in $H'$ as two copies of the same fixed finite graph; likewise for every edge $e\in E(\ball{T+1}{G_1}{x})$ the edge gadget $\egadget{e}{\ssigma{E_1}{}(e)}$ appears in both graphs as two copies of the same fixed finite graph.

For each such node gadget (resp.\ edge gadget), fix an isomorphism between its copy in $G'_1$ and its copy in $H'$. Since gadgets are vertex-disjoint, these gadget-wise isomorphisms together define a bijection $\varphi : S \to V(H')$ by setting $\varphi(v)$ to be the image of $v$ under the isomorphism of the (unique) gadget containing $v$. This defines an isomorphism between the induced subgraph $G'_1[S]$ and the corresponding induced subgraph of $H'$. Indeed, $\varphi$ preserves edges inside each gadget by construction; and it preserves edges between
gadgets because inter-gadget edges are created by the encoding solely according to incidences in $G_1$, and $H$
contains exactly the incidences among vertices/edges of $\ball{T+1}{G_1}{x}$.

Now let $a\in V(G'_1)$ be any vertex in the node gadget of $x$, and define $a_H := \varphi(a)$.
Since $\ball{T}{G'_1}{a}$ can only reach gadgets corresponding to vertices/edges within distance at most $T+1$ from $x$
in $G_1$, we have $\ball{T}{G'_1}{a}\subseteq G'_1[S]$. Therefore, the restriction of $\varphi$ to
$\ball{T}{G'_1}{a}$ is a bijection onto $\ball{T}{H'}{a_H}$ and preserves adjacency proving the isomorphism.
\end{proof}

\subsection{Simulation of algorithms}\label{sec:simulation_of_algorithms}

Given an encoding-decoding scheme that follows the framework of \cref{sec:graph_encoding,sec:graph_decoding}, let $G'$ be a $3$-regular graph and its decoding be  $(G,\ssigma{V(G)}{},\ssigma{E(G)}{})$ where $G$ belongs to a subfamily of graphs with degree at most $d$ and $\ssigma{V(G)}{}:V(G)\rightarrow [k_1]$ and $\ssigma{E(G)}{}:E(G)\rightarrow [k_2]$ for some $k_1,k_2>0$ are the decoded node and edge labelings of $G$. Let $G''$ be the perfect  encoding of  $(G,\ssigma{V(G)}{},\ssigma{E(G)}{})$. Observe that the way the encoding-decoding framework is defined, $G'$ and $G''$ are isomorphic only when $G'$ is a perfect encoding and does not have any malformed parts. 

Let there be algorithms $\mathcal{A}_1$ and $\mathcal{A}_2$ such that the following holds for some integer $k>0$.
\begin{enumerate}
	\item $\mathcal{A}_1$ runs on $(G,\ssigma{V(G)}{},\ssigma{E(G)}{})$ producing an output labeling $\ssigma{G}{out} : V(G)\rightarrow [k]$ and has locality~$T_1$.
	\item $\mathcal{A}_2$ runs on $G''$ producing an output labeling $\ssigma{G''}{out} : V(G'')\rightarrow [k]$ and has locality~$T_2$.
\end{enumerate}

We now define a lift of the output labelings from $G$ and $G''$ to $G'$ and $G$ respectively as follows.

\paragraph{\boldmath Lifting $\ssigma{G}{out}$ to $G'$.}
We define a labeling $\ssigma{G'}{out} : V(G')\rightarrow [k]$ as follows.

\begin{enumerate}
	\item  For each node $v\in V(G)$ and for each node $u$ in $\decode{v}{\ssigma{V(G)}{}(v)}$, we have that $\ssigma{G'}{out}(u) = \ssigma{G}{out}(v)$.
	\item Every other node $u'$ in $G'$ is assigned an arbitrary but fixed default value $\bot\in [k]$, i.e. $\ssigma{G'}{out}(u') = \bot$. 
\end{enumerate}

\paragraph{\boldmath Lifting $\ssigma{G''}{out}$ to $G$.}
The labeling $\ssigma{G}{out}:V(G)\rightarrow [k]$ is defined as follows : for every node $v\in V(G)$ and an arbitrary node $u\in \ngadget{v}{\ssigma{V(G)}{}(v)}$, $\ssigma{G}{out}(v) = \ssigma{G''}{out}(u)$. Furthermore, we make the following assumption on $\ssigma{G''}{out}$ that every vertex in every node-gadget is assigned the same output label and every node in every edge-gadget is assigned a fixed but arbitrary label $\bot\in [k]$.

These lifts can be computed by simulating algorithms $\mathcal{A'}_1$ and $\mathcal{A'}_2$ with locality $T_1' = \Lambda T_1$ and $T_2' = T_2$ respectively which simulate $\mathcal{A}_1$ and $\mathcal{A}_2$ respectively. In the next few subsections, we describe the simulations in various models. 

\subsubsection{Simulation in \local}
\paragraph{\boldmath Description of $\mathcal{A'}_1$.}
The algorithm $\mathcal{A'}_1$  running on node $v\in V(G')$ first applies \local decoding procedure (\cref{lemma:local_decoding_procedure}). If $v$ belongs to an edge-gadget, $\ssigma{G'}{out}(v)$ is set as fixed but an arbitrary default value $\bot\in [k]$.  Otherwise, it belongs to $\decode{u}{\ssigma{V(G)}{}(u)}$ for some node $u\in V(G)$. If it is not the representative node, then it remains idle for the rest of the simulation rounds and finally takes in the output label of the representative node. Else, $v$ simulates one communication round of the $\mathcal{A}_1$ on node $u$ and  node $w\in N_G(u)$ via the path $P_{w:v,x_w}$ (produced by \cref{lemma:local_decoding_procedure}) where $x_w$ is the representative node of $\decode{w}{\ssigma{V(G)}{}(w)}$. Since  the length of $P_{w:v,x_w}$ is bounded by $\Lambda$, one communication round of $\mathcal{A}_1$ is simulated by $\mathcal{A'}_1$ in at most $\Lambda$ rounds and therefore, in $O(\Lambda T)$ rounds, node $v$ simulates $T$ rounds of $\mathcal{A}_1$ on $u$ and sets $\ssigma{G'}{out}(v) \coloneqq \ssigma{G}{out}(u)$ and finally broadcasts it to all other nodes of $\decode{u}{\ssigma{V(G)}{}(u)}$.

\paragraph{\boldmath Description of $\mathcal{A'}_2$.}
The algorithm $\mathcal{A'}_2$ running on node $v\in V(G)$  collects \[(\ball{T_2}{G}{v},\rsigma{V(G)}{}{\ball{T_2}{G}{v}},\rsigma{E(G)}{}{\ball{T_2}{G}{v}})\] and applies graph encoding on it to extract $\ball{T_2}{G''}{a}$ where the node $a$ is an arbitrary vertex in $\ngadget{v}{\ssigma{V(G)}{}(v)}$ and simulates $\mathcal{A}_2$ and sets $\ssigma{G}{out}(v) = \ssigma{G''}{out}(a)$. 
\subsubsection{Simulation in \slocal}
\paragraph{Description of $\mathcal{A'}_1$.}
Let  $p$ be the adversarial order of vertices  of $G'$ that the algorithm $\mathcal{A'}_1$ processes. This implicitly defines an order $p'$ of vertices in $G$ as follows: for each vertex $v_i\in V(G')$ accessed according to the ordering of $p$, if $v_i\in \decode{u}{\ssigma{V(G)}{}(u)}$ and $\{v_1,\ldots, v_{i-1}\}\cap\decode{u}{\ssigma{V(G)}{}(u)}=\varnothing$, then insert $u$ in $p'$. Let $p = (v_1,\ldots, v_{|V(G')|})$ and $p' = (u_1,\ldots,u_{|V(G)|})$.

The algorithm $\mathcal{A'}_1$ processing the nodes of $G'$ according to $p$ is equivalent to simulating $\mathcal{A'}_1$ on $G$ according to $p'$. While processing vertex $v_i$, the algorithm $\mathcal{A'}_1$ first checks if $v_i$ belongs to an edge gadget and sets $\ssigma{G'}{out}(v_i) = \bot$ if it is. Else, it applies \cref{cor:local_decoding_radius_T} to get \[\bigl( \ball{T}{G}{u_j}, \rsigma{V(G)}{}{\ball{T}{G}{u_j}}, \rsigma{E(G)}{}{\ball{T}{G}{u_j}} \bigr),\] where $v_i\in \decode{u_j}{\ssigma{V(G)}{}(u_j)}$.
Also, for each $x\in \ball{T}{G}{u_j}\cap \{u_1,\ldots,u_j\}$, algorithm $\mathcal{A'}_1$ sets $\ssigma{G}{out}(x) = \ssigma{G'}{out}(y)$ where $y\in \{v_1,\ldots,v_i\}$ such that $y\in \decode{x}{\ssigma{V(G)}{}(x)}$. It then simulates algorithm $\mathcal{A}_1$ on  $\bigl( \ball{T}{G}{u_j}, \rsigma{V(G)}{}{\ball{T}{G}{u_j}}, \rsigma{E(G)}{}{\ball{T}{G}{u_j}} \bigr)$ and the partial function $\ssigma{G}{out}$ restricted to $\{u_1,\ldots,u_j\}$ and computes $\ssigma{G}{out}(u_j)$ and sets $\ssigma{G'}{out}(v_i) = \ssigma{G}{out}(u_j)$.

\paragraph{Description of $\AA'_2$.} Let $p=(u_1,u_2,\dots, u_{|V(G)|})$ be the adversarial order of vertices for $G$ that algorithm $\AA_2'$ processes. The algorithm works as follows: following the order provided by the adversary, it processes one node at the time. While processing node $u_i$, the algorithm $\AA_2'$ collects the radius $T_2$ neighborhood around $u_i$ in $G$ and expands it to the relative radius $T'$ ball in $G'$ and selects a representative gadget node $v_i^r$ (as the main computing unit) among the $j$ gadgets $v_1,\dots v_j$ representing $u_i$ in $G'$ in addition, it lifts all the $\sigma^{\oupt}_G(w)$ for all nodes $w$ in its $T_2$ radius neighborhood in $G'$ whose out labels have been already set. Subsequently, it checks whether there is already another gadget node $v_j$ (among the gadget nodes representing $u_i$) in $G'$ for which their states have already been set to an output label. If that is the case, node $u_i$ sets its own output label $\sigma_{G}^{\oupt}(u_i)$ equal to that one. Otherwise, the node $u_i$ simulates $\AA_2$   on that representative node in $G'$ and sets its output label $\sigma_G^{\oupt}(u_i) = \sigma_{G'}^{\oupt}(v_i^r)$.

\subsubsection{Simulation in \nonsign model}

\paragraph{\boldmath Description of $\mathcal{A}'_1$.}
We assume there is a $T_1$-non-signaling distribution $\outcome(G)$ for $(G,\ssigma{V(G)}{},\ssigma{E(G)}{})$.

	\subparagraph{Construction of the outcome.} For graph $G'$,  define $\outcome(G')$ by the following sampling process:
	\begin{itemize}[noitemsep]
			\item[1)] Sample $\sigma_{G}^{\oupt}\sim \outcome(G)$
			\item[2)] Lift $\ssigma{G}{out}$ to $\ssigma{G'}{out}$. 
	\end{itemize}

	\subparagraph{Non-signaling argument.}
	We prove that $\outcome(G')$ is non-signaling beyond $T_1'=\Lambda T_1$.
	Let $A\subseteq V(G')$ be any node set, and consider another $3$-regular graph $G'_1$ such that $A'\subseteq V(G'_1)$ 
	such that
	\begin{align}\label{eq:isomorphic_view}
		\View_{T_1'}^{G'}(A)\cong \View_{T_1'}^{G'_1}(A').
	\end{align}
	We show that $\outcome(G')[A]\cong \outcome(G'_1)[A']$.
	Let $(G_1,\ssigma{V(G_1)}{},\ssigma{E(G_1)}{})$ be the decoded graph of $G'_1$.
Define $\tilde{A} \coloneqq \{u\in V(G)\ \mid\ \decode{u}{\ssigma{V(G)}{}(u)}\cap A\neq \varnothing\}$ and $\tilde{A'}$ is defined in the same way w.r.t $G_1$.
	By applying \cref{cor:local_decoding_radius_T,eq:isomorphic_view}, we have that 

	\[
		\View_{T}^{G}(\tilde{A})\cong \View_T^{G_1}(\tilde{A'}).
	\]
	and by our assumption, $\outcome(G)[\tilde{A}]\cong \outcome(G_1)[\tilde{A'}]$

		By the definition of the lift, we have that for each node $w\in A$ such that $w\in \decode{x}{\ssigma{V(G)}{}(x)}$ for some node $x\in V(G)$, $\ssigma{G'}{out}(w)= \ssigma{G}{out}(x)$ and every other node in $A$ is labeled $\bot$ and the same happens for each node $y\in A'$.  Since the lifted labels of $A$ only depends on $\ssigma{G}{out}(\tilde{A})$ and therefore any further decoded structure is irrelevant for the marginal on $A$, and therefore,  $\outcome(G')[A]\cong \outcome(G'_1)[A']$.

\paragraph{\boldmath Description of $\mathcal{A}'_2$.}

We assume there is a $T_2$ non-signaling outcome for $G''$.

	\subparagraph{Construction of the outcome.} For graph $G$,  define $\outcome(G)$ by the following sampling process:
	\begin{itemize}[noitemsep]
			\item[1)] Sample $\sigma_{G''}^{\oupt}\sim \outcome(G'')$
			\item[2)] Lift $\ssigma{G''}{out}$ to $\ssigma{G}{out}$. 
	\end{itemize}
	
We prove that $\outcome(G)$ is non-signaling beyond $T_2$. Let $B\subseteq V(G)$ be any node set and consider another labeled graph $(G_1,\ssigma{V(G_1)}{},\ssigma{E(G_1)}{})$ such that there exists a set $B_1\subseteq V(G_1)$ such that 
\begin{align}\label{eq:iso2}
	\View_{T_2}^G(B)\cong \View_{T_2}^{G_1}(B_1)
\end{align}
We show that $\outcome(G)[B]\cong \outcome(G_1)[B_1]$.

Let $G'_1$ be the encoding of $(G_1,\ssigma{V(G_1)}{},\ssigma{E(G_1)}{})$.
Define $\tilde{B} \coloneqq \bigcup\limits_{x\in B}\decode{x}{\ssigma{V(G)}{}(x)}$ and  $\tilde{B_1}$ is defined in the same way w.r.t $G_1$.  Since~\eqref{eq:iso2} holds, by \cref{lem:locality_of_encoding}, we have that $\View_{T_2}^{G''}(\tilde{B})\cong \View_{T_2}^{G'_1}(\tilde{B_1})$ and by our assumption $\outcome(G'')[\tilde{B}]\cong \outcome(G'_1[\tilde{B_1}])$. By the definition of the lift and the assumption on the output labelings, we have that $\outcome(G)[B]\cong \outcome(G_1)[B_1]$.

\subsubsection{Simulation in \findept model}

\paragraph{\boldmath Simulation of $\AA'_1$.}
We assume that there is a $T_1$-dependent outcome $\outcome(G)$ for the decoded instance
$(G,\sigma_V(G),\sigma_E(G))$.

\subparagraph{Construction of the outcome.}
For a 3-regular graph $G'$, define $\outcome(G')$ by:
\begin{itemize}[noitemsep]
		\item Sample $\sigma^{\oupt}_G \sim \outcome(G)$.
	\item Lift $\sigma^{\oupt}_G$ to $\sigma^{\oupt}_{G'}$.
\end{itemize}

\subparagraph{\boldmath $T_1'$-dependence argument.}
Let $A,B\subseteq V(G')$ be any two sets with $\dist_{G'}(A,B)>T_1'$.
Define
\begin{align*}
	\widetilde{A} &\coloneqq \{u\in V(G)\ \mid\ \decode{u}{\ssigma{V(G)}{}(u)}\cap A\neq \varnothing\},\\
 \widetilde{B} &\coloneqq \{u\in V(G)\ \mid\ \decode{u}{\ssigma{V(G)}{}(u)}\cap B\neq \varnothing\}.
\end{align*}
Then $\dist_G(\widetilde{A},\widetilde{B})>T_1$.
Since $\outcome(G)$ is $T$-dependent, $\sigma^{\oupt}_G[\widetilde{A}]$ and $\sigma^{\oupt}_G[\widetilde{B}]$ are independent.
By definition of the lift, $\sigma^{\oupt}_{G'}[A]$ is a deterministic function of
$\sigma^{\oupt}_G[\widetilde A]$ (and the fixed instance $G'$), and similarly
$\sigma^{\oupt}_{G'}[B]$ is a deterministic function of $\sigma^{\oupt}_G[\widetilde B]$.
Therefore $\sigma^{\oupt}_{G'}[A]$ and $\sigma^{\oupt}_{G'}[B]$ are independent, i.e., $\outcome(G')$ is
$T'$-dependent.

\paragraph{\boldmath Simulation of $\AA_2'$.}

We assume that there is a $T_2$-dependent outcome $\outcome(G'')$ for the perfect encoding~$G''$.

\subparagraph{Construction of the outcome.} For the decoded labeled graph $(G,\sigma_V(G),\sigma_E(G))$,
define $\outcome(G)$ by
\begin{itemize}[noitemsep]
		\item Sample $\sigma^{\oupt}_{G''}\sim \outcome(G'')$.
		\item Define $\sigma^{\oupt}_G$ by lifting $\sigma^{\oupt}_{G''}$ to $G$.
\end{itemize}

\subparagraph{\boldmath $T_2$-dependence argument.}
Let $B,C\subseteq V(G)$ with $\dist_G(B,C)>T_2$.
Define the corresponding gadget-sets in $G''$ by $\widetilde{B} \coloneqq \bigcup\limits_{x\in B}\decode{x}{\ssigma{V(G)}{}(x)}$ and $\widetilde{C} := \bigcup_{x\in C}\decode{x}{\ssigma{V(G)}{}(x)}$.
In a perfect encoding, contracting each node gadget of $G''$ to a single vertex yields (a copy of) $G$,
so distances cannot decrease under this contraction. Hence
\[
\dist_{G''}(\widetilde B,\widetilde C)\ \ge\ \dist_G(B,C)\ >\ T_2.
\]
Since $\outcome(G'')$ is $T_2$-dependent, the random labelings
$\sigma^{\oupt}_{G''}[\widetilde{B}]$ and $\sigma^{\oupt}_{G''}[\widetilde{C}]$ are independent.
By construction, $\sigma^{\oupt}_G[B]$ is a deterministic function of $\sigma^{\oupt}_{G''}[\widetilde{B}]$, and similarly for $\sigma^{\oupt}_G[C]$. Therefore $\sigma^{\oupt}_G[B]$ and $\sigma^{\oupt}_G[C]$ are independent.
This holds for all $B,C$ with $\dist_G(B,C)>T_2$, so $\outcome(G)$ is $T_2$-dependent.

\subsubsection{Simulation in \qlocal}

\paragraph{\boldmath Description of $A'_1$.}
The description of the algorithm is identical to the one for \local with the difference that: (1) the model is message passing, (2) the representative node runs a quantum algorithm, and (3) the edge-gadget nodes lying on the designated path for transmitting qubits apply a $\SWAP$ operator to forward qubits on such a path.

\paragraph{\boldmath Description of $\AA_2'$.}
Each node $u\in V(G)$ locally constructs a virtual copy of the constant-size gadget
vertices that it should take care of: all vertices in the node gadget of $u$, and for every incident edge
$\{u,w\}$ a deterministically chosen part (or all of it) of the corresponding edge gadget, so that every gadget vertex of the encoded structure is owned by exactly one endpoint.
In every round, node $u$ applies to its owned virtual gadget vertices exactly the same unitaries and measurements that $\AA_1'$ would apply at those vertices. Whenever $\AA_1'$ would transmit qubits across an edge whose endpoints are owned by different nodes $u$ and $w$, node $u$ sends those qubits directly to $w$ over the physical edge $\{u,w\}$ in one round. Thus, one round in $\AA_2'$ corresponds to one round in $\AA_1'$. This implies that the number of rounds needed by $\AA_2'$ is the same as $\AA_1'$ up to a constant additive factor needed for the decoding step.

\subsection{Symmetry breaking and above-symmetry-breaking complexity}\label{sec:symmetry_breaking}

By \emph{symmetry breaking} we refer to the task of finding a distance-$k$ vertex coloring with $O(1)$ colors, for a constant $k$, in bounded-degree graphs. Note that here the number of colors can depend on distance $k$ and maximum degree $\Delta$.

Recall that distance-$k$ coloring is the canonical task that allows us to solve all LCL problems that have locality $O(\log^* n)$ in the deterministic \local model: for any such LCL problem $\Pi$, there is a constant $k$ such that an oracle that finds a distance-$k$ coloring, followed by constant-locality post-processing, is sufficient to solve $\Pi$ \cite{chang-kopelowitz-pettie-2019-an-exponential-separation}. Furthermore, being able to find distance-$k$ coloring for \emph{some} $k$ easily translates to the ability to find distance-$k$ coloring for \emph{any} $k$, through a simple local reduction (in essence, find a coloring in the power graph).

For every model \(\mathcal{M}\) in our definitions, we define a function \(\diamondsuit(\mathcal{M})\) that is the symmetry-breaking complexity in this model. That is, asymptotic locality \(\diamondsuit(\mathcal{M})\) is sufficient and necessary to find a distance-$k$ coloring in model \(\mathcal{M}\). We recall the well-known results:
\begin{itemize}
    \item In the det-\local and rand-\local models, the symmetry-breaking complexity is $\Theta(\log^* n)$: we can use the Cole--Vishkin technique \cite{cole-vishkin-1986-deterministic-coin-tossing-with} to break symmetry, and this is optimal \cite{linial-1992-locality-in-distributed-graph-algorithms,naor-1991-a-lower-bound-on-probabilistic-algorithms-for}.
    \item In the \slocal and \olocal models, the symmetry-breaking complexity is $O(1)$: the sequential processing order trivially breaks symmetry (we can find a coloring greedily).
    \item In the \nonsign and bounded-dependence models, the symmetry-breaking complexity is $O(1)$. This is non-trivial, but we can exploit finitely dependent colorings from \cite{holroyd-liggett-2016-finitely-dependent-coloring} and its adaptation to any symmetry-breaking LCL from \cite{akbari-coiteux-roy-etal-2025-online-locality-meets}.
    \item Finally, in the \qlocal model, the symmetry-breaking complexity is still a major open problem, but it is known to be bounded by $O(\log^* n)$.
\end{itemize}
In particular, \(\diamondsuit(\mathcal{M})\) is always bounded by $O(\log^* n)$, and hence as long as we consider complexities in the range $\Omega(\log^* n)$, we can ignore additive \(\diamondsuit(\mathcal{M})\) terms.

We will use symmetry-breaking in this work for finding, in particular, distance-$k$ vertex coloring and distance-$k$ edge coloring, for various constants $k$.

\section{From LCL with inputs to LCL without inputs}\label{sec:eliminate-inputs}

In this section, we describe a reduction from an arbitrary promise-free LCL $A = (\mathcal{G}_d, \inlabels_A,\outlabels_A, \mathcal{C}_A, r_A)$ to a promise-free LCL without inputs $B =(\mathcal{G}_3, \outlabels_B, \mathcal{C}_B,r_B)$ where $\mathcal{G}_3$ denotes the family of $3$-regular graphs  preserving the locality of the problem across various models.



\subsection{Graph encoding}

\paragraph*{Ladder-Gadget.}

We denote the \emph{ladder-gadget} of length $k$ (for $k\geq 0$)  by $\mathcal{H}_k$ which we describe below.

\begin{description}
    \item[Vertices.]
    $V(\mathcal{H}_k) = \{a_0,\ldots,a_k,\, b_0,\ldots,b_k,\, x, t, u, v, w\}$.

    \item[Ladder structure.]
    The sequences of vertices $(a_0, a_1, \ldots, a_k, t, u)$ and $(b_0, b_1, \ldots, b_k, w, v)$ form two parallel paths.
    For every $0 \le i \le k$, add an edge $\{a_i, b_i\}$, forming the rungs of the ladder.
    Finally, the vertices $\{t, u, v, w\}$ induce the graph $K_4 \setminus \{t, w\}$.

    \item[Triangle.]
    Connect $x$ to both $a_0$ and $b_0$, forming a triangle on $\{x, a_0, b_0\}$.
\end{description}

Attaching $\mathcal{H}_k$ to a vertex $v$ of another graph means adding the edge $\{v, x\}$.


\paragraph{Description of Encoding.}
Given a graph $G$ and an input function $\sigma:V(G)\rightarrow [k]$ for some $k\geq 0$, its \emph{graph encoding} is a transformation of $G$ into a $3$-regular graph $G'$  by replacing each vertex and edge of $G$ as follows.

The function $\ssigma{}{}$ naturally can be interpreted as a pair of two functions $(\ssigma{V(G)}{}),\ssigma{E(G)}{}$ where $\ssigma{V(G)}{} = \ssigma{}{}$, for each edge $\{u,v\}\in E(G)$, $\ssigma{E(G)}{}(\{u,v\}) = (\sigma(u),\sigma(v))$. Now, we proceed with the encoding of nodes and edges.

\begin{enumerate}
	\item{\textbf{Encoding of nodes.}} For each vertex $v\in V(G)$, replace it with a \emph{node gadget} $(C_v,\sigma(v))$ which consists of a chordless cycle $C_v = (v_1,\ldots, v_{\deg_{G}(v)+2})$ and two copies of $\mathcal{H}_{\sigma(v)}$ attached to vertices $v_1$ and $v_2$. 
	\item{\textbf{Encoding of edges.}} For each edge $\{u,v\}\in E(G)$, add an \emph{edge gadget} $e_{u,v}$ that connects a designated vertex of $C_u$ to a designated vertex of $C_v$. The gadget $e_{u,v}$ is a $3$-path $\{u^\star,a_u^\star, a_v^\star, v^\star\}$, where $u^\star\in V(C_u)$ and $v^\star\in V(C_v)$ are the chosen vertices on the two cycles, and the auxiliary vertices $a_u^\star$ and $a_v^\star$ are attached with gadgets $\sigma(u)$ and $\sigma(v)$ respectively. The attachment vertices $u^\star$ and $v^\star$ on the cycles are chosen so that every vertex in the resulting graph $G'$ has degree exactly $3$.
   \end{enumerate}


\begin{lemma}[Uniqueness of Encoding]\label{lemma:unique_encoding}
        Let $G_1$ and $G_2$ be two graphs, and let $\sigma_1 : V(G_1) \to [k]$ and $\sigma_2 : V(G_2) \to [k]$ be input labelings for some $k \ge 0$. Let $G_1'$ and $G_2'$ denote the respective encodings of $(G_1, \sigma_1)$ and $(G_2, \sigma_2)$. Then, if $G_1' \cong G_2'$, it must hold that  $(G_1,\ssigma{1}{})$ and $(G_2,\ssigma{2}{})$  are also isomorphic.%
\end{lemma}

\begin{proof}
	Let the isomorphism between $G'_1$ and $G'_2$ be $\varphi$.
		Let $D_1 = \{(v,i)\in V(G_1)\times [k]\ \mid\ \sigma_1(v) = i \}$ and $D_2 = \{(u,j)\in V(G_2)\times [k]\ \mid\ \sigma_2(u) = j\}$.
	We define a mapping $\varphi_2 : D_1\rightarrow D_2$ as follows:
	\[
			\text{For each node $v\in V(G_1)$, } \varphi_2(v,\sigma_1(v)) = \mathsf{decode}(\varphi(\decode{v}{\ssigma{V(G_1)}{}(v)})).
	\]
We next show that $\varphi_2$ is an isomorphism between the decoded labeled graphs:
\begin{description}
\item[Label preservation.]
Let $(v,i)\in D_1$ and let $\varphi_2(v,i)=(u,j)$.
By construction, the node gadget $X_v := \decode{v}{\ssigma{V(G_1)}{}(v)}$ contains two attached copies of
$\mathcal H_i$, and the value $i$ is uniquely determined from the isomorphism type of $X_v$.
Since $\varphi$ is a graph isomorphism, $\varphi(X_v)$ is a node gadget in $G_2'$ isomorphic to $X_v$,
hence it also contains two copies of $\mathcal H_i$. Therefore
$\mathsf{decode}(\varphi(X_v))=(u,i)$, and thus $j=i$.
\item[Edge preservation.]
Let $v,w\in V(G_1)$ be distinct.
By definition of the encoding, $\{v,w\}\in E(G_1)$ if and only if there exists an edge gadget in $G_1'$
connecting a designated vertex of the cycle $C_v$ inside $X_v$ to a designated vertex of the cycle $C_w$ inside $X_w$.
As $\varphi$ preserves adjacency, it maps this edge gadget to an edge gadget in $G_2'$ connecting the cycles
inside $\varphi(X_v)$ and $\varphi(X_w)$.
Let $\varphi_2(v,\sigma_1(v))=(u,\sigma_2(u))$ and $\varphi_2(w,\sigma_1(w))=(z,\sigma_2(z))$.
Then $\varphi(X_v)$ and $\varphi(X_w)$ are exactly the node gadgets corresponding to $u$ and $z$ in $G_2'$,
so the existence of such an edge gadget is equivalent to $\{u,z\}\in E(G_2)$.
Hence $\{v,w\}\in E(G_1)\iff \{u,z\}\in E(G_2)$.
\end{description}
Combining label preservation and edge preservation, $\varphi_2$ induces an isomorphism
$(G_1,\sigma_1)\cong (G_2,\sigma_2)$, i.e.\ $(G_1,\ssigma{1}{})\cong (G_2,\ssigma{2}{})$.
\end{proof}
\subsection{Problem description}
\paragraph{Problem $B$.}\label{def:problem:b}
Problem~$B$ is a promise-free LCL problem without inputs, specified as follows.
\begin{itemize}[noitemsep]
    \item $\mathcal{G}_B$ is the family of all $3$-regular graphs.
    \item $\outlabels_B = \outlabels_A$.
	\item The checkability radius is $r_B = \Lambda\;r_A$.
	\item The constraints $\mathcal{C}_B$ are described below.
\end{itemize}


An element of $\mathcal{C}_B$ is a pair $(G(w),\ssigma{G}{out})$ where $G(w)$ is a $3$-regular graph centered at $w$ and  of radius at most $r_B$ and $\ssigma{G}{out}:V(G(w))\to \outlabels_B$ is an output labeling.

		Let $(G',\ssigma{G'}{in},\ssigma{G'}{out})$ be the decoding of $(G(w),\ssigma{G}{out})$. We have $(G(w),\ssigma{G}{out})\in \mathcal{C}_B$ if and only if \emph{one} of the following two cases holds.

\begin{enumerate}
	\item The node $w\in\decode{u}{\ssigma{G'}{in}(u)}$ for some node $u\in V(G')$. Then the following holds.
		\[
			\bigl(\ball{r_A}{G'}{u}, \rsigma{G'}{in}{\ball{r_A}{G'}{u}}, \rsigma{G'}{out}{\ball{r_A}{G'}{u}}\bigr)\in \mathcal{C}_A
		\]
	\item Otherwise, $\ssigma{G}{out}(w)$ can take any arbitrary value from $\outlabels_{B}$.
\end{enumerate}
    	


Without loss of generality, we assume that every pair $(G(w),\ssigma{G}{out}) \in \mathcal{C}_B$ has radius exactly $r_B$. Instances of smaller radius can be handled separately, as they only arise from input graphs of constant diameter.
\begin{lemma}[Soundness of the Lift from $A$ to$B$]\label{lemma:sound_lift_a_b}
	Let $\instance_A = (G_A,\sigma_{A}^{in})$ be an instance of problem $A$. Let $\instance_B = (G_B)$ be the corresponding encoded instance for $B$. Let $\sigma_{B}^{out}$ be any legal output labeling on $\instance_B$. Define an output labeling $\sigma_A^{out}$ for $G_A$ by assigning $\sigma^{out}_A(u) = \sigma_{B}^{out}(v)$ where $\sigma_{B}^{out}(v)$  where $v\in\decode{u}{\ssigma{A}{in}(u)}$. Then $\ssigma{A}{out}$ is a valid output labeling for $\instance_A$.
\end{lemma}
\begin{proof}
	Fix any vertex $u\in V(G_A)$, let $v$ be one of the nodes that are part of the gadget associated to node $u$ in $B$, i.e., $v\in  \decode{u}{\ssigma{V(G_A)}{in}(u)}$. Since $\sigma_B^{out}$ is legal for problem $B$,  the neighborhood satisfies the constraint $\CC_B$, formally
\[
\bigl(\ball{r_B}{G_B}{v},\rsigma{G_B}{out}{\ball{r_B}{G_B}{v}})\in \mathcal{C}_B.
\]
$G_B$ is correctly gadgeted in nodes and edge gadgets. Let $(G'(w),\ssigma{G'}{in},\ssigma{G'}{out})$ be the decoding of $	\bigl(\ball{r_B}{G_B}{v},\rsigma{G_B}{out}{\ball{r_B}{G_B}{v}})$ where $v\in \decode{w}{\ssigma{G'}{in}(w)}$, by the definition of problem $B$ we have that the centered graph at a node $w$ of radius at most $r_B$ satisfies $\CC_B$ if and only if

\[
\bigl(\ball{r_A}{G'}{w}, \rsigma{G'}{in}{\ball{r_A}{G'}{w}}, \rsigma{G'}{out}{\ball{r_A}{G'}{w}}\bigr)\in \CC_A.
\]
Applying \cref{cor:local_decoding_radius_T} we have that 
\[
\bigl(\ball{r_A}{G'}{w}, \rsigma{G'}{in}{\ball{r_A}{G'}{w}}, \rsigma{G'}{out}{\ball{r_A}{G'}{w}}\bigr)\cong	\bigl(\ball{r_A}{G_A}{u}, \rsigma{G_A}{in}{\ball{r_A}{G_A}{u}}, \rsigma{G_A}{out}{\ball{r_A}{G_A}{u}}\bigr)\in \CC_A.
\]
Thus $\sigma_{A}^{out}$ is a valid output labeling for $A$.
\end{proof}
\begin{lemma}[Soundness of the Lift from $B$ to $A$]\label{lemma:sound_lift_b_a}
	Let $\instance_B=(G_B)$ be an instance of problem $B$, and let $\instance_A=(G_A,\sigma_A^{in})$ be the instance for the problem $A$, achieved by decoding node and edge gadgets in $G_B$. Let $\sigma^{out}_A$ be a legal output labeling for problem $A$ on $\instance_A$. Define the output labeling for all $v\in V(G_B)$ as
	\[\sigma_B^{out}(v)=\begin{cases}
		\sigma_{A}^{out}(u) &\text{if $v\in\decode{u}{\ssigma{V(G_A)}{in}(u)}$} \text{ for some }u\in V(G_A) \\
		\bot &\text{otherwise.}
	\end{cases}
	\]
	Then $\sigma_B^{out}$ is legal for problem $B$ on $G_B$.

\end{lemma}
\begin{proof}
	Given $\instance_B= (G_B)$ we apply the decoding of $G_B$ and obtain $\instance_A = (G_A,\sigma_{A}^{in})$. After that, we apply the labeling $\sigma_{A}^{out}$ and lift the output labels to $\instance_B$. Let $v\in V(G_B)$ be a node such that $v\in \decode{u}{\ssigma{V(G_A)}{in}(u)}$ for some $u\in V(G_A)$ (otherwise, by the second point of the description of problem $B$, $\ssigma{B}{out}(v)$ is legal trivially). Consider the ball of radius $r_B$ around $v$
			\[
		\bigl(\ball{r_B}{G_B}{v},\rsigma{G_B}{out}{\ball{r_B}{G_B}{v}})
		\]
		and let $(G'(w),\ssigma{G'}{in},\ssigma{G'}{out})$ be its decoding.  Now, consider ball of  radius $r_A$ around $w$ in $G'$ 
		\[
		\bigl(\ball{r_A}{G'}{w}, \rsigma{G'}{in}{\ball{r_A}{G'}{w}}, \rsigma{G'}{out}{\ball{r_A}{G'}{w}}\bigr)
		\]
		and consider the $r_A$ radius neighborhood  of a node $u$ in $G_A$
		\[
		\bigl(\ball{r_A}{G_A}{u}, \rsigma{G_A}{in}{\ball{r_A}{G_A}{u}}, \rsigma{G_A}{out}{\ball{r_A}{G_A}{u}}\bigr).
		\]
		By \cref{cor:local_decoding_radius_T}, we have that 
		\[
			\bigl(\ball{r_A}{G'}{w}, \rsigma{G'}{in}{\ball{r_A}{G'}{w}}, \rsigma{G'}{out}{\ball{r_A}{G'}{w}}\bigr)	\cong \bigl(\ball{r_A}{G_A}{u}, \rsigma{G_A}{in}{\ball{r_A}{G_A}{u}}, \rsigma{G_A}{out}{\ball{r_A}{G_A}{u}}\bigr)\in \CC_A.
		\]
		Therefore, by the first point of the  description of problem $B$, we have that 
		
		\[
		\bigl(\ball{r_B}{G_B}{v},\rsigma{G_B}{out}{\ball{r_B}{G_B}{v}})\in \CC_B
		\]
\noindent
	thus $\sigma_B^{out}$ is a valid labeling for $B$.		
\end{proof}

\begin{lemma}\label{lemma:prob_a_to_b}
	Let $\instance_A = (G_A,\sigma_A^{in})$ be an instance of problem $A$ Let $G_B$ be the corresponding encoded instance for $B$. Let $\outcome_B(\instance_B)=\{(\sigma_B^{\oupt,i},p_i)\}_{i\in I}$ be an outcome that solves $B$ on $\instance_B$ with probability at least $q$. Define $\outcome_A(\instance_A)$ by sampling $(\sigma_{B}^{out,i},p_i)$ from $\outcome_B(\instance_B)$ and applying the $\LIFT$. If the lift is sound, then $\outcome(\instance_A)$ solves $A$ with probability at least $q$.
\end{lemma}
\begin{proof}
	Let
	\begin{align*}
		S_B &:= \{ i\in I \mid \sigma_B^{(\oupt,i)} \text{ is legal for $B$ on } G_B\}, \\
		S_A &:= \{ i\in I \mid \sigma_A^{(\oupt,i)} \text{ is legal for $A$ on } \instance_A\}.
	\end{align*}
	By soundness of the lift from $B$ to $A$, for every $i\in I$,
	\[
	i\in S_B \Rightarrow  \LIFT(\sigma_B^{(\oupt,i)}) \text{ is legal for $A$ on } \instance_A
	\Rightarrow i\in S_A.
	\]
	Hence $S_B\subseteq S_A$. Therefore,
	\[
	\sum_{i\in S_A} p_i  \geq \sum_{i\in S_B} p_i \geq q,
	\]
	where the last inequality is the assumption that $\outcome_B$ solves $B$ with success probability at least $q$.
\end{proof}
\begin{lemma}\label{lemma:prob_b_to_a}
Let $\instance_B = (G_B)$ be an instance of problem $B$ and let $\instance_A = (G_A,\sigma_A^{in})$ be the instance for problem $A$ achieved by decoding node and edge gadgets in $G_B$. Let $\outcome_A(\instance_A) = \{(\sigma_{A}^{(\oupt,i)},p_i)\}_{i\in I}$ be any outcome that solves $A$ on $\instance_A$ with probability at least $q$. Define the outcome $\outcome_B(\instance_B)$ by sampling $(\sigma_{A}^{(\oupt,i)},p_i)$ from $\outcome_A(\instance_A)$ and applying the $\LIFT$. If the $\LIFT$ is sound, then $\outcome_B(\instance_B)$ solves $B$ with probability at least $q$. 
\end{lemma}
\begin{proof}
	Let $S_A := \{ i\in I \mid \sigma_A^{(\oupt,i)} \text{ solves } A \text{ on } \instance_A\}$ and
	$S_B := \{ i\in I \mid \sigma_B^{(\oupt,i)} \text{ solves } B \text{ on } G_B\}$.
	By soundness of the lift from $A$ to $B$, for every $i\in I$,
	\[
	i\in S_A \Rightarrow\  \LIFT(\sigma_A^{(\oupt,i)}) \text{ is legal for } B
	\Rightarrow\ i\in S_B.
	\]
	Hence $S_A\subseteq S_B$, and therefore
	\[
	\sum_{i\in S_B} p_i \ \ge\ \sum_{i\in S_A} p_i \ \ge\ q,
	\]
	where the last inequality is exactly the assumption that $\outcome_A$ solves $A$ with probability at least $q$.
	Thus $\outcome_B$ solves $B$ with probability at least $q$.
\end{proof}

\begin{lemma}\label{lemma:simulation_a_b}
	If problem $A$ admits an optimal algorithm with locality $\Theta(T) = \Omega(\diamondsuit(\MM))$ for a model $\MM$, then problem $B$ admits an algorithm with locality $\Theta(T)$.
\end{lemma}
\begin{proof}
	First, let us prove that such algorithm with locality $\Theta(T)$ exists. Let $G_B$ be an instance of problem $B$. The algorithm $\AA_B$ will simulate $\AA_A$  as in \cref{sec:framework} on the decoded version of $G_B$. By the soundness of the $\LIFT$ (\cref{lemma:sound_lift_a_b,lemma:sound_lift_b_a}) the resulting $\sigma_{B}^{\oupt}$ satisfies problem $B$. The simulation has $\Theta(T)$ locality. Thus, the locality of $\AA_B$ is $\Theta(T)$. Now, let us prove that the locality of $\AA_B$ is $\Omega(T)$. For the sake of contradiction, assume it is $o(T)$. Let us take an instance $\instance_A= (G_A,\sigma_{A}^{in})$ for problem $A$. 
	Then simulate $\AA_B$ on the decoded $G_B$ as it is described in \cref{sec:framework}. The algorithm will have $o(T)$ locality and by the soundness of the $\LIFT$ (\cref{lemma:sound_lift_a_b,lemma:sound_lift_b_a}) the resulting $\sigma_{A}^{\oupt}$ is valid. This means that the resulting algorithm for problem $A$ takes $o(T)$, but we assumed that the optimal algorithm for $A$ had locality $\Theta(T)$, thus it would be not optimal hence such an $o(T)$-locality algorithm cannot exist.
\end{proof}

\begin{remark}[Comparing randomized algorithms]
	We compare randomized algorithms by \((T,q)\) (locality, success probability).
	We say that \(\alg{A}\) is at least as good as \(\alg{B}\) if \(T_A\le T_B\) and \(q_A\geq q_B\)
	(with one strict for ``strictly better''). An algorithm is optimal if no other algorithm is
	strictly better in this sense.
\end{remark}

\section{From an LCL without inputs to a PN-checkable problem}

In this section, we describe a reduction from an arbitrary always-solvable LCL without inputs \(B =(\mathcal{G}_3, \outlabels_B, \mathcal{C}_B,r_B)\) to an always-solvable PN-checkable LCL without inputs \(D =(\mathcal{G}_3, \outlabels_D, \mathcal{C}_D,r_D)\). We will start by introducing graph encoding along with relevant theorem statements. Then we will prove the reduction and, finally, we will return to the graph-theoretical theorems and provide proofs for them.

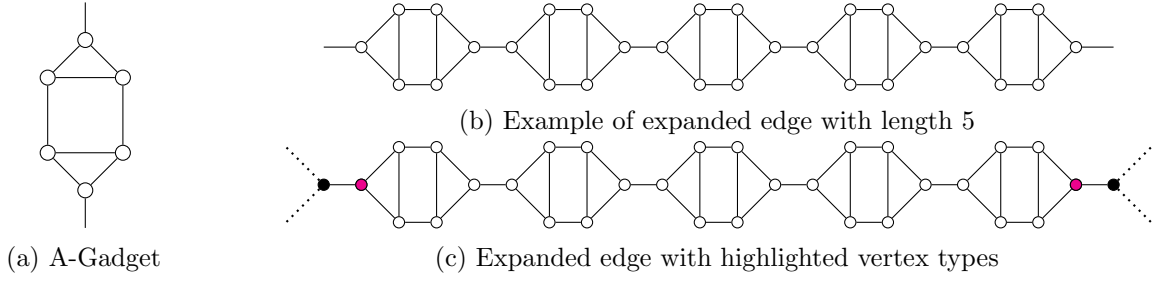
\begin{figure}
    \centering
    \begin{subfigure}{0.25\textwidth}
        \centering

        \begin{tikzpicture}[scale=0.5, every node/.style={circle, draw, inner sep=2pt}]
            \node (1) at (0,2) {};
            \node (2) at (-1,1) {};
            \node (3) at (1,1) {};
            \node (4) at (-1,-1) {};
            \node (5) at (1,-1) {};
            \node (6) at (0,-2) {};
        
            \draw (0, 3) -- (1);
            \draw (1) -- (2);
            \draw (1) -- (3);
            \draw (2) -- (3);
            \draw (2) -- (4);
            \draw (3) -- (5);
            \draw (4) -- (5);
            \draw (4) -- (6);
            \draw (5) -- (6);
            \draw (6) -- (0, -3);
        \end{tikzpicture}

        \caption{A-Gadget}
        \label{fig:gadget}
    \end{subfigure}
    \hfill
    \begin{subfigure}{0.73\textwidth}
        \begin{subfigure}[b]{\textwidth}
            \centering

        \begin{tikzpicture}[scale=0.5, every node/.style={circle, draw, inner sep=1.5pt, font=\small}]
        
        \newcommand{\gadget}[2]{
          \node (1#2) at #1 {};
          \node (2#2) at ($(1#2) + (1,1)$) {};
          \node (3#2) at ($(1#2) + (1,-1)$) {};
          \node (4#2) at ($(1#2) + (2,1)$) {};
          \node (5#2) at ($(1#2) + (2,-1)$) {};
          \node (6#2) at ($(1#2) + (3,0)$) {};
          
          \draw (1#2) -- (2#2) -- (4#2) -- (6#2) -- (5#2) -- (3#2) --(1#2);
          \draw (2#2) -- (3#2);
          \draw (4#2) -- (5#2);
        }
        
        \gadget{(0,0)}{A}
        \gadget{(4,0)}{B}
        \gadget{(8,0)}{C}
        \gadget{(12,0)}{D}
        \gadget{(16,0)}{E}

        \draw (1A) -- ($(1A) + (-1, 0)$);
        \draw (6E) -- ($(6E) + (1, 0)$);
        
        \foreach \X/\Y in {A/B, B/C, C/D, D/E}{
          \draw (6\X) -- (1\Y);
        }
        
        \end{tikzpicture}
        
        \caption{Example of expanded edge with length 5}
        \label{fig:expedge}
        \end{subfigure}
        
        \begin{subfigure}[b]{\textwidth}
            \centering

        \begin{tikzpicture}[scale=0.5, every node/.style={circle, draw, inner sep=1.5pt, font=\small}]
        
        \newcommand{\gadget}[2]{
          \node (1#2) at #1 {};
          \node (2#2) at ($(1#2) + (1,1)$) {};
          \node (3#2) at ($(1#2) + (1,-1)$) {};
          \node (4#2) at ($(1#2) + (2,1)$) {};
          \node (5#2) at ($(1#2) + (2,-1)$) {};
          \node (6#2) at ($(1#2) + (3,0)$) {};
          
          \draw (1#2) -- (2#2) -- (4#2) -- (6#2) -- (5#2) -- (3#2) --(1#2);
          \draw (2#2) -- (3#2);
          \draw (4#2) -- (5#2);
        }
        
        \gadget{(0,0)}{A}
        \gadget{(4,0)}{B}
        \gadget{(8,0)}{C}
        \gadget{(12,0)}{D}
        \gadget{(16,0)}{E}

        \node[fill=black] (F) at ($(1A) + (-1, 0)$) {}; 
        \draw (1A) -- (F);
        \draw[dotted, thick] (F) -- ($(F) + (-1, 1)$);
        \draw[dotted, thick] (F) -- ($(F) + (-1, -1)$);
        \node[fill=black] (L) at ($(6E) + (1, 0)$) {};
        \draw (6E) -- (L);
        \draw[dotted, thick] (L) -- ($(L) + (1, 1)$);
        \draw[dotted, thick] (L) -- ($(L) + (1, -1)$);
        
        \foreach \X/\Y in {A/B, B/C, C/D, D/E}{
          \draw (6\X) -- (1\Y);
        }

        \node[fill=magenta] at (1A) {};
        \node[fill=magenta] at (6E) {};
        
        \end{tikzpicture}
        
        \caption{Expanded edge with highlighted vertex types}
        \label{fig:expedgecol}
        \end{subfigure}
    \end{subfigure}
    \caption{Illustrations for gadget expansion}
    \label{fig:gadget_expedge}
\end{figure}

\subsection{Graph encoding}\label{sec:graph-encoding}

In this subsection, we introduce the graph encoding that will be used, together with relevant definitions, lemmas, and theorems. 
We will first introduce the following concepts:
\begin{description}
\item[A-gadget.]
We will define A-gadget as a gadget with vertices \(V = \{1, 2, 3, 4, 5, 6\}\) and edges \(E = \{\{1\}, \{1, 2\}, \{1, 3\}, \{2, 3\}, \{2, 4\}, \{3, 5\}, \{4, 5\}, \{5, 6\}, \{4, 6\}, \{6\}\}\) (as illustrated on the \cref{fig:gadget}). 
\item[Expanded A-edge.]
Let us take an edge \(e\) with label \(w \in \mathbb{N}\) and replace it with a sequence of \(w\) A-gadgets.  We will call this an \emph{expanded edge}. We will refer to the number of gadgets in the expanded edge as the \emph{length}. An example of an expanded edge with length 5 is illustrated in \cref{fig:expedge}.
\item[Original vertex.]
An original vertex is a vertex that was a part of the original graph and thus is not a part of any gadget. These vertices are colored black on \cref{fig:expedgecol}.
\item[Outer vertex.]
An outer vertex is a vertex that belongs to a gadget and is connected to an original vertex. So, every original vertex would be connected to three outer vertices. These vertices are colored pink on \cref{fig:expedgecol}.
\item[Inner vertex.]
An inner vertex is a vertex that belongs to a gadget and is not connected to any original vertex. These vertices are colored white on \cref{fig:expedgecol}.
\item[Gadgeting.]
First, let us take a 3-regular simple graph \((G, x)\), where \(x: E(G) \to \mathbb{N}\) is a distance-\(2r\) edge coloring. Then \(\text{gadget}(G, x) = (G', w)\), where \(w: E(G') \to \{0, 0.5\}\) and each edge in \(E(G)\) is expanded by definition. Finally, \(w\) is defined as follows:
\begin{itemize}[noitemsep]
    \item \(w(i, j) = 0.5 \text{ if }i \text{ is original vertex or } j \text{ is original vertex}\)
    \item \(w(i, j) = 0 \text{ otherwise}\)
\end{itemize}
We will refer to the resulting graph as \emph{correctly gadgeted 3-regular graph}.
\item[Gadgeted distance.]
We will define function \(\text{gdist}(u, v) = \min_{\pi \in \Pi}\sum_{v_i \in \pi}w(v_i, v_{i + 1})\).
\begin{remark}
    Note, that for two original vertices \(u, v \in G\) with distance \(d\), that after expansion became vertices \(u', v' \in G': \text{gdist}(u', v') = d\).
\end{remark}
\item[Gball.]
The \(r\)-gball for a vertex \(v\) includes all vertices and edges at gdist \(\leq r\) from \(v\). We will denote it as \(\gn_r(v)\).
\end{description}

Now we are ready to state the main result of this section, which we will prove in \cref{ssec:graph-theoretic-proofs}:

\begin{theorem}[Coloring of gadgeted graphs]\label{thm:gadgeting} Let \(r\) be  a positive natural number. Assume that \(G\) is a correctly gadgeted graph with labeling \(\chi: V(G) \to \mathbb{N}\) that satisfies:
\begin{enumerate}
	\item \(\chi\) is a distance 2 coloring: for every vertex \(v\) and vertices \(w, u \in \mathcal{N}(v) \cup {v}:\) \(\chi(w) \neq \chi(u)\).
    \item \(\chi\) is edge-consistent in every \(r\)-gball: if there exists \(\{v, u\} \in E(G)\), then for every vertex \(v' \in \gn_{r}(v)\) such that \(\chi(v') = \chi(v)\) there exists a vertex \(u'\) such that \(\{v', u'\} \in E(G)\) and \(\chi(u') = \chi(u)\).
\end{enumerate}
Then \(\chi\) is a correct gdistance-\(2r\) original coloring, that is, for each original vertex \(v\) it holds that every vertex in \(\gn_r(v)\) has a unique color.
\end{theorem}

We will also prove the following lemma in \cref{ssec:graph-theoretic-proofs}:

\begin{restatable}[Original vertex is not a part of a triangle]{lemma}{outerdisc}\label{lem:outer_disconnected}
    Let \(v\) be an original vertex with color \(\chi(v)\). Let \(w_1, w_2, w_3\) be its neighbors. Then there is no pair of vertices \(\{u_1, u_2\} \in E(G)\) in the $r$-gball of \(v\), that \(\chi(u_1) = \chi(w_i)\) and \(\chi(u_2) = \chi(w_j)\), where \(i \neq j \text{ and } i, j \in \{1, 2, 3\}\). Or, in other words, there is no edge between the vertex with the same color, as two outer vertices.
\end{restatable}

Finally, let us show that the encoding is unique:

\begin{lemma}[Uniqueness of encoding]\label{lem:unique_encoding}
        Let $G_1$ and $G_2$ be two graphs, and let $x_1 : E(G_1) \to [k]$ and $x_2 : E(G_2) \to [k]$ be input edge-labelings for some $k \ge 0$.
	Let $G_1'$ and $G_2'$ denote the respective encodings of $(G_1, x_1)$ and $(G_2, x_2)$. Then, if $G_1' \cong G_2'$, it must hold that  $(G_1,x_1)$ and $(G_2,x_2)$  are also isomorphic.%
\end{lemma}
\begin{proof}
	Let us denote \(\zeta: V(G) \to V(G')\) as an injective function between vertices of graph \((G, x)\) and its encoding \(G'\). Such that \(\forall v \in V(G):\ \zeta(v)\) is original, and \(\forall \{v, w\} \in E(G), x(\{v, w\}) = m\) exists an expanded edge of length \(m\) between \(\zeta(v)\) and \(\zeta(w)\).
	Let the isomorphism between $G'_1$ and $G'_2$ be $\varphi$.
        We define a mapping $\varphi_2 : V(G_1)\rightarrow V(G_2)$ as follows:
        \[
	\text{For each node $v\in V(G_1)$, } \varphi_2(v) = \mathsf{decode}(\varphi(\mathsf{decode}(v))).
        \]
We next show that $\varphi_2$ is an isomorphism between the decoded labeled graphs.

Let $v,w\in V(G_1)$ be distinct.
By definition of the encoding, $\{v,w\}\in E(G_1)$ and \(x_1(\{v, w\}) = m\) if and only if there exists an expanded edge in $G_1'$
connecting a vertex \(\zeta(v)\) and \(\zeta(w)\).
As $\varphi$ preserves adjacency, it maps this expanded-edge to an expanded-edge
of the same length in $G_2'$.
Let $\varphi_2(v)=v'$ and $\varphi_2(w)=w'$.
Then $\varphi(\zeta(v))$ and $\varphi(\zeta(w))$ are exactly the nodes corresponding to $v'$ and $w'$ in $G_2'$,
so the existence of such expanded-edge is equivalent to $\{v',w'\}\in E(G_2)$ and \(x_2(\{v', w'\}) = m\).
Hence $\{v,w\}\in E(G_1) \text{ and } x_1(\{v, w\}) = m \iff \{v',w'\}\in E(G_2) \text{ and } x_2(\{v', w'\}) = m$.
\end{proof}

\subsection{PN-decoding}
\paragraph{Non-triangle vertex.} We will call a vertex \emph{non-triangle} if it is not a part of any triangle.

\begin{remark}
	Note that if distance-$2$ neighborhood of vertex \(u\) is included in \(T^r(v)\), then we can easily see if \(u\) is a non-triangle vertex or not.
\end{remark}

\begin{remark}
	By \cref{lem:outer_disconnected}, in a correctly gadgeted 3-regular graph only original vertices are non-triangle vertices. 
\end{remark}

\paragraph{Non-triangle restricted view.}
\(\check{T}^r(v)\) is a subgraph of \(T^r(v)\), such that the root \(v\) is a non-triangle vertex and every path from the 
root \(v\) to the leaves includes exactly \(r_B\) non-triangle vertices, and the last vertex in the path is a non-triangle vertex. 
\begin{remark}
    \(\check{T}^r(v)\) does not exist for some \(T^r(v)\). For example, if vertex \(v\) is not a non-triangle vertex.
\end{remark}

\paragraph{Correctly gadgeted balls.}  Let \(\mathcal{H}_g^r(v)\) be the set of all possible $r$-gballs of correctly gadgeted graphs.

\paragraph{First reconstruction function for D.}\label{def:r:d1} Let  $(\check{T}^r(v),\ssigma{D}{out})$ be a non-triangle restricted PN-view, the reconstruction function $R_{D1}(\cdot)$ is
\[R_{D1}(\check{T}^r(v), \ssigma{D}{out}), \]
defined as follows:
\begin{itemize}
                \item For every non-empty set \(M_\beta = \left\{u | u \in V(\check{T}^r(v)), \ssigma{D}{out}[0](u) = \beta\right\}\) there exists exactly one \(u' \in V(H^r(v))\). Let it be \(u' = \gamma(M_\beta)\).
                \item For every edge \((u, w) \in E(T^r_G(v))\) there exists exactly one edge: \[(u', w') = (\gamma(M_{\chi(u)}), \gamma(M_{\chi(w)})) \in E(H^r(v)).\]
                \item There are no other edges or vertices besides specified above.
\end{itemize}

\paragraph{Second reconstruction function for D.}\label{def:r:d2} 
 Let  $(\check{T}^r(v),\ssigma{D}{out})$ be a non-triangle restricted PN-view, such that \(R_{D1}\left(\check{T}^r(v), \ssigma{D}{out}\right) \in \mathcal{H}_g^r(v)\), the reconstruction function $R_{D2}(\cdot)$ is
\[R_{D2}(\check{T}^r(v), \ssigma{D}{out}) = H^r(v), \]
defined as follows:
\begin{itemize}
                \item For every vertex \(u\) such that \( u \in V(\check{T}^r(v)) \text{ and }u \text{ is non-triangle} \) there exists exactly one \(u' \in V(H^r(v))\). Let it be \(u' = \gamma(u)\).
                \item \((u', w') = (\gamma(u), \gamma(w)) \in E(H^r(v))\) if and only if there exists a path between \(u\) and \(w\), consisting of only triangle vertices (not non-triangle vertices).
                \item There are no other edges or vertices besides specified above.
                \item For all $u' = \gamma(u) \in H^r(v)$ we have $\ssigma{B}{out}(u') = \ssigma{D}{out}[1](u)$.
\end{itemize}

\paragraph{Problem D.}\label{def:problem:d}
Problem D is an LCL problem defined on $3$-regular graphs as follows:
\begin{itemize}[noitemsep]
    \item \(r_D = (4 k + 1)  r_B + 1\), where \(k\) is a constant that is larger than the number of edges in any $3$-regular centered graph of radius \(r_B\),
    \item \(\inlabels_D = \emptyset\),
    \item \(\outlabels_D = [\zeta] \times \outlabels_B\), where \(\zeta \) is an arbitrary constant that is larger than the maximum number of vertices in any 3-regular graph of radius \(r_D\), and
    \item \(\ssigma{D}{out} = \ssigma{D}{out}[0] \times \ssigma{D}{out}[1] \).
\end{itemize}
We say that \(f: V \to [\zeta]\) is reconstructible if it is a distance-$2$ coloring and is edge consistent, meaning that for every two vertices that are labeled the same, the labels of their neighbors are also the same. Let \(I(f)\) be an indicator function that returns True iff \(f\) is reconstructible.
Finally, we define the family of centered PN-views as follows:
\begin{align*}
	\mathcal{C}_D \coloneqq{}
	&\bigl\{ \left(T^r(v), \ssigma{D}{out}\right) \mid I(\ssigma{D}{out}[0]) \wedge \exists \check{T}^r(v) \wedge R_{D1}\left(\check{T}^r(v), \ssigma{D}{out}\right) \in \mathcal{H}_g^r(v) \wedge R_{D2}\left(\check{T}^r(v), \ssigma{D}{out}\right) \in \mathcal{C}_B \bigr\} \\
	{}\vee{}
	&\bigl\{ \left(T^r(v), \ssigma{D}{out}\right) \mid  I(\ssigma{D}{out}[0])  \wedge  \left( \nexists \check{T}^r(v) \vee R_{D1}\left(\check{T}^r(v), \ssigma{D}{out}\right) \notin \mathcal{H}_g^r(v) \right) \bigr\}.
\end{align*}

\subsection{Correctness of the encoding}

\begin{lemma}[Soundness of the Lift from $B$ to $D$]\label{lem:lift:b_to_d}
	 Let $\instance_B=(G_B)$ be an instance of problem $B$, let \(x\) be any distance \(2r_B\) edge coloring and let $\instance_D = (G_D)$ be the
	corresponding encoded (gadgeted) instance for problem $D$.
	Let $\ssigma{D}{out}$ be any legal output labeling for $\instance_D$.
	Define an output labeling $\ssigma{B}{out}$ for $G_B$ by identifying each vertex
	$u\in V(G_B)$ with its corresponding \emph{original} (non-triangle) vertex \(v\) in $G_D$ and setting
	\[
	\ssigma{B}{out}(u) := \ssigma{D}{out}[1](v).
	\]
	Then \(\ssigma{B}{out}\) is legal for problem \(B\) on \(G_B\).
\end{lemma}
\begin{proof}
	Fix any vertex $u\in V(G_B)$, and let $v$ be the corresponding original vertex in $G_D$.
	Since $\ssigma{D}{out}$ is legal for problem $D$ it satisfies the constraint $\CC_D$.
	Consider the radius $r_D$ PN-view $(T^{r_D}_{G_D}(v),\ssigma{D}{out})$.
	Since $G_D$ is correctly gadgeted, the only non-triangle vertices are the original vertices.
	Thus, the non-triangle restricted view \((\check{T}^{r_D}(v), \ssigma{D}{out})\) exists for an original vertex \(v\).
	By definition of problem D, \(\ssigma{D}{out}[0]\) is reconstructible in \((T^{r_D}_{G_D}, \ssigma{D}{out})\) and thus is a distance 2 coloring and edge consistent in $(\check{T}^{r_D}(v), \ssigma{D}{out})$.
	By \cref{thm:gadgeting} it implies that \(\ssigma{D}{out}[0]\) is a correct gdistance \(2r_B\) original coloring.
	Meaning, every vertex in the gball of original vertex \(v\) has a unique color. 
	Thus, \(R_1\) function works correctly and \(R_1(\check{T}^{r_D}(v), \ssigma{D}{out})\) belongs to $\mathcal{H}^{r_D}_g(v)$.
	By definition of problem \(D\) it implies that \(R_2(\check{T}^{r_D}(v), \ssigma{D}{out})\) must belong to \(\mathcal{C}_B\).
	
	Finally, \(R_2(\check{T}^{r_D}(v), \ssigma{D}{out}) \cong (B^{r_B}_{G_B}(u), \ssigma{B}{out})\) by definition of gadgeting and \(R_2\). Thus, for \(\forall u \in V(G_B): (B^{r_B}_{G_B}(u), \ssigma{B}{out}) \in \mathcal{C}_B\) and \(\ssigma{B}{out}\) is legal for problem \(B\) on \(G_B\).   
\end{proof}

\paragraph{Contracted graph.} Let \(G_D\) be an instance of problem \(D\). Then we will define the \emph{contracted} \(G_D\) as follows. Locally, we contract every expanded edge by replacing it with a single edge. If the result contains any multi-edges or loops, we reverse contraction on those edges.

\begin{lemma}[Soundness of the Lift from $D$ to $B$]\label{lem:lift:d_to_b}
	Let $\instance_D=(G_D)$ be an instance of problem $D$, and let $G_B$ be the
	instance for the problem \(B\), achieved by contracting \(G_D\).
		Let $\ssigma{B}{out}$ be a legal output labeling for problem \(B\) on $G_B$.
Let \(\chi\) be any distance \(2r_D\) coloring. 		
Define an output labeling $\ssigma{D}{out} = (\ssigma{D}{out}[0], \ssigma{D}{out}[1])$ as $\ssigma{D}{out}[0] = \chi$, and 
		\[\ssigma{D}{out}[1](v) =
                \begin{cases}
			\ssigma{B}{out}(u) & \text{if $u$ is a match in \(G_B\) for \(v\)},\\
                        \bot & \text{otherwise.}
		\end{cases}\]
		Then $\ssigma{D}{out}$ is legal for problem \(D\)  on $G_D$.
\end{lemma}

\begin{proof}
Fix any vertex $v\in V(G_D)$ and consider the labeled radius $r_D$ PN-view
	$\bigl(T^{r_D}_{G_D}(v),\ssigma{D}{out}\bigr)$.
By construction, $\ssigma{D}{out}[0]=\chi$ is reconstructible, hence
	$I_{\mathrm{rec}}(\ssigma{D}{out}[0])$ holds. 
		Now, if the non-triangle restricted view of \(\bigl(T^{r_D}_{G_D}(v),\ssigma{D}{out}\bigr)\) does not exist, 
	then the view is legal by the second rule of problem
	\(D\). If \(\bigl(\check{T}^{r_D}_{G_D}(v),\ssigma{D}{out}\bigr)\) exists
	but \(R_1\bigl(\check{T}^{r_D}_{G_D}(v),\ssigma{D}{out}\bigr) \notin \mathcal{H}^{r_D}_g\), then the view is again legal by the second rule of problem \(D\).
	Finally, if \(\bigl(\check{T}^{r_D}_{G_D}(v),\ssigma{D}{out}\bigr)\) exists
	and \(R_1\bigl(\check{T}^{r_D}_{G_D}(v),\ssigma{D}{out}\bigr) \in \mathcal{H}^{r_D}_g\), then, since \(\ssigma{D}{out}[0]\) is a distance \(2r_D\) coloring, \(R_2\bigl(T^{r_D}_{G_D}(v),\ssigma{D}{out}\bigr) \cong (B^{r_B}_{G_B}(u), \ssigma{B}{out})\). Hence, \(R_2\bigl(T^{r_D}_{G_D}(v),\ssigma{D}{out}\bigr) \in \mathcal{C}_B\).
	Thus, \(\forall v\in V(G_D) : \bigl(T^{r_D}_{G_D}(v),\ssigma{D}{out}\bigr) \in \mathcal{C}_D\) and the \(\ssigma{D}{out}\) is legal for problem \(D\) on \(G_D\).
\end{proof}
\begin{lemma}\label{lemma:prob_b_to_d}
	Let $\instance_B =(G_B)$ be an instance of problem $B$. Let $G_D$ be the corresponding encoded (gadgeted) instance for problem $D$. Let $\outcome_D(\instance_D) = \{(\sigma_{D}^{(out,i)},p_i)\}_{i\in I}$ be an outcome that solves $D$ on $\instance_D$ with probability at least $q$. Define $\outcome_B(\instance_B)$ by sampling $(\sigma_{B}^{(out,i)},p_i)$ from $\outcome_D(\instance_D)$ and applying the $\LIFT$. If the lift is sound, the $\outcome(\instance_B)$ solves $B$ with probability at least $q$.
\end{lemma}
\begin{lemma}\label{lemma:prob_d_to_b}
		 Let $\instance_D =(G_D)$ be an instance of problem $D$. Let $G_B$ be the instance for the problem $B$ achieved by contracting every expanding edge in $G_D$. Let $\outcome_B(\instance_B)= \{(\sigma_{B}^{(out,i)},p_i)\}_{i\in I}$  be any outcome that solves $B$ on $G_B$ with probability at least $q$. Define $\outcome_D(\instance_D)$ by sampling $(\sigma_{B}^{(out,i)},p_i)$ from $\outcome_B(G_B)$ and applying the $\LIFT$. If the lift is sound, then $\outcome_D(\instance_D)$ solves $D$ with probability at least $q$.
\end{lemma}
The proofs of \cref{lemma:prob_b_to_d,lemma:prob_d_to_b} follow from the soundness of the lift (\cref{lem:lift:b_to_d,lem:lift:d_to_b}) and use the same probabilistic arguments as in \cref{lemma:prob_a_to_b,lemma:prob_b_to_a}, respectively.

\begin{lemma}\label{lemma:simulation_b_d}
	If problem $B$ admits an optimal algorithm with locality $\Theta(T) = \Omega(\diamondsuit(\mathcal{M}))$ for a model \(\mathcal{M}\), then problem $D$ admits an algorithm with locality $\Theta(T)$.
\end{lemma}

\begin{proof}
	First, let us prove that such algorithm with locality \(\Theta(T)\) exists. Let \(G_D\) be an instance of problem \(D\). 
	The algorithm \(\alg{D}\) will consist of two parts: coloring and simulation. First, we will provide a valid distance \(2r_D\) coloring, which by definition has locality \(O(\diamondsuit(\mathcal{M}))\).
	Then, we will simulate the \(\alg{B}\) as in \cref{sec:framework} on contracted version of \(G_D\). By the soundness of lift \cref{lem:lift:d_to_b} the resulting \(\ssigma{D}{out}\) satisfies problem \(D\). The simulation has \(\Theta(T)\) locality. Thus, locality of \(\alg{D}\) is \(O(\diamondsuit(\mathcal{M})) + \Theta(T) = O(T) + \Theta(T) =  \Theta(T)\).

	Now, let us prove that locality of \(\alg{D}\) is \(\Omega(T)\). Assume, it is not, then its locality is \(o(T)\). Let us take $G_B$ an instance of problem \(B\). We can get a distance \(2r_B\) edge-coloring in \(\Theta(\diamondsuit(\mathcal{M}))\). 
	Then we can simulate \(\alg{D}\) on a gadgeted \(G_B\) as it is described in \cref{sec:framework}. 
	It will have \(o(T)\) locality and by \cref{lem:lift:b_to_d} resulting \(\ssigma{B}{out}\) is valid. 
That implies that the resulting algorithm for problem \(B\) takes \(\Theta(\diamondsuit(\mathcal{M})) + o(T)\). 

Assume, \(\alg{B}\) has locality \(\omega(\diamondsuit(\mathcal{M}))\), then we arrive at contradiction as \(\Theta(\diamondsuit(\mathcal{M})) + o(T) = o(T)\) and \(\alg{B}\) is not optimal. Now, assume \(\alg{B}\) has locality \(\Theta(\diamondsuit(\mathcal{M}))\). Since, as a part of the output (\(\ssigma{D}{out}[0]\)) \(\alg{D}\) must provide a valid coloring, \(\alg{D} = \Omega(\diamondsuit(\mathcal{M}))\). Thus, \(\alg{D}\) cannot have locality \(o(T)\).
\end{proof}

\subsection{Graph-theoretical proofs}\label{ssec:graph-theoretic-proofs}

Now, we proceed with proving the already introduced graph-theoretical result \cref{thm:gadgeting}. To do so, we first prove a few auxiliary lemmas. We start with the following lemma we already introduced:

\outerdisc*

\begin{proof}

    W.l.o.g., assume that \(\chi(v) = 1\), \(\chi(w_1) = 2\), \(\chi(w_2) = 3\), \(\chi(w_3) = 4\). For contradiction, w.l.o.g.\ assume that there exists \(\{u_1, u_2\} \in E(G)\), such that \(\chi(u_1) = \chi(w_1) = 2\) and \(\chi(u_2) = \chi(w_3) = 4\). By 2nd assumption of \cref{thm:gadgeting} we know, that if such edge exists, then exists \(u \in \mathcal{N}(w_3)\) such that \(\chi(u) = \chi(w_1) = 2\). Thus, the labeling \(\chi\) of our graph will look like \cref{fig:lemma4:b} modulo symmetry.

    Following the same logic, \(u\) must be connected to a vertex labeled with \(\chi(v) = 1\). Then this vertex must be located as
    represented in blue in \cref{fig:lemma4:c,fig:lemma4:d,fig:lemma4:e}.

    Now, let us proceed with proving why every possible location of vertex labeled with 1 will lead to a contradiction.

    Let us start with placing it as in \cref{fig:lemma4:c}. We can easily notice that this creates a contradiction, as vertex labeled with 4 is connected to two vertices with label 1. So \(\chi\) is not a distance-2-coloring, which contradicts the 1st assumption of \cref{thm:gadgeting}.

    Next is placing vertex labeled 1, as in \cref{fig:lemma4:d,fig:lemma4:e}. This would imply that its remaining neighbors are labeled 3 and 4 per 2nd assumption of \cref{thm:gadgeting}. Thus, vertices labeled 4 and 3 are connected. Thus, the remaining vertex in the gadget must be labeled 3, as every vertex labeled with 4 must contain a vertex labeled with 3 in its neighborhood.

    Now we can notice, that both placements of 3 and 4 (\cref{fig:lemma4:d,fig:lemma4:e}) are invalid, as neither of them is a valid distance 2 coloring and thus both contradict the assumption of \cref{thm:gadgeting}.

    By exhaustion of all options, none of the neighbors of \(w_3\) can be labeled with 2. Thus this edge doesn't exist and by the symmetry this edge doesn't exist for any other pair of outer vertices.
\end{proof}

\begin{figure}[ht]
    \centering
    \begin{subfigure}{0.19\textwidth}
        \centering
        \begin{tikzpicture}[scale=0.8, every node/.style={circle, draw, inner sep=2pt}]
            \node (2) at (-1,4) {2};
            \node (3) at (1,4) {3};
            \node (1) at (0,3) {1};
            \node (4) at (0,2) {4};
            \node (5) at (-1,1) {?};
            \node (6) at (1,1) {?};
            \node (7) at (-1,-1) {?};
            \node (8) at (1,-1) {?};
            \node (9) at (0,-2) {?};
        
            \draw (1) -- (2);
            \draw (1) -- (3);
            \draw (1) -- (4);
            \draw (4) -- (5);
            \draw (4) -- (6);
            \draw (5) -- (6);
            \draw (5) -- (7);
            \draw (6) -- (8);
            \draw (7) -- (8);
            \draw (8) -- (9);
            \draw (7) -- (9);
        \end{tikzpicture}
        \caption{}
        \label{fig:lemma4:a}
    \end{subfigure}
    \begin{subfigure}{0.19\textwidth}
        \centering
        \begin{tikzpicture}[scale=0.8, every node/.style={circle, draw, inner sep=2pt}]
            \node (2) at (-1,4) {2};
            \node (3) at (1,4) {3};
            \node (1) at (0,3) {1};
            \node (4) at (0,2) {4};
            \node (5) at (-1,1) {2};
            \node (6) at (1,1) {?};
            \node (7) at (-1,-1) {?};
            \node (8) at (1,-1) {?};
            \node (9) at (0,-2) {?};
        
            \draw (1) -- (2);
            \draw (1) -- (3);
            \draw (1) -- (4);
            \draw (4) -- (5);
            \draw (4) -- (6);
            \draw (5) -- (6);
            \draw (5) -- (7);
            \draw (6) -- (8);
            \draw (7) -- (8);
            \draw (8) -- (9);
            \draw (7) -- (9);
        \end{tikzpicture}
        
        \caption{}
        \label{fig:lemma4:b}
    \end{subfigure}
    \begin{subfigure}{0.19\textwidth}
        \centering
        \begin{tikzpicture}[scale=0.8, every node/.style={circle, draw, inner sep=2pt}]
            \node (2) at (-1,4) {2};
            \node (3) at (1,4) {3};
            \node (1) at (0,3) {1};
            \node (4) at (0,2) {4};
            \node (5) at (-1,1) {2};
            \node (6) at (1,1) {\textcolor{cyan}{1}};
            \node (7) at (-1,-1) {?};
            \node (8) at (1,-1) {?};
            \node (9) at (0,-2) {?};
        
            \draw (1) -- (2);
            \draw (1) -- (3);
            \draw[magenta, very thick] (1) -- (4);
            \draw (4) -- (5);
            \draw[magenta, very thick] (4) -- (6);
            \draw (5) -- (6);
            \draw (5) -- (7);
            \draw (6) -- (8);
            \draw (7) -- (8);
            \draw (8) -- (9);
            \draw (7) -- (9);
        \end{tikzpicture}
        \caption{}
        \label{fig:lemma4:c}
    \end{subfigure}
    \hfill
    \begin{subfigure}{0.19\textwidth}
        \centering
        \begin{tikzpicture}[scale=0.8, every node/.style={circle, draw, inner sep=2pt}]
            \node (2) at (-1,4) {2};
            \node (3) at (1,4) {3};
            \node (1) at (0,3) {1};
            \node (4) at (0,2) {4};
            \node (5) at (-1,1) {2};
            \node (6) at (1,1) {\textcolor{magenta}{3}};
            \node (7) at (-1,-1) {\textcolor{cyan}{1}};
            \node (8) at (1,-1) {\textcolor{orange}{3}};
            \node (9) at (0,-2) {\textcolor{orange}{4}};
        
            \draw (1) -- (2);
            \draw (1) -- (3);
            \draw (1) -- (4);
            \draw (4) -- (5);
            \draw (4) -- (6);
            \draw (5) -- (6);
            \draw (5) -- (7);
            \draw[magenta, very thick] (6) -- (8);
            \draw (7) -- (8);
            \draw (8) -- (9);
            \draw (7) -- (9);
        \end{tikzpicture}
        \caption{}
        \label{fig:lemma4:d}
    \end{subfigure}
    \hfill
    \begin{subfigure}{0.19\textwidth}
        \centering
        \begin{tikzpicture}[scale=0.8, every node/.style={circle, draw, inner sep=2pt}]
            \node (2) at (-1,4) {2};
            \node (3) at (1,4) {3};
            \node (1) at (0,3) {1};
            \node (4) at (0,2) {4};
            \node (5) at (-1,1) {2};
            \node (6) at (1,1) {\textcolor{magenta}{3}};
            \node (7) at (-1,-1) {\textcolor{cyan}{1}};
            \node (8) at (1,-1) {\textcolor{orange}{4}};
            \node (9) at (0,-2) {\textcolor{orange}{3}};
        
            \draw (1) -- (2);
            \draw (1) -- (3);
            \draw (1) -- (4);
            \draw (4) -- (5);
            \draw (4) -- (6);
            \draw (5) -- (6);
            \draw (5) -- (7);
            \draw[magenta, very thick] (6) -- (8);
            \draw (7) -- (8);
            \draw[magenta, very thick] (8) -- (9);
            \draw (7) -- (9);
        \end{tikzpicture}
        \caption{}
        \label{fig:lemma4:e}
    \end{subfigure}
    \hfill
    \caption{Illustrations for \cref{lem:outer_disconnected}. Issues highlighted in pink, potential placements of 1 are in blue and implied placements of 3 and 4 are in orange. }
    \label{fig:lemma4}
\end{figure}
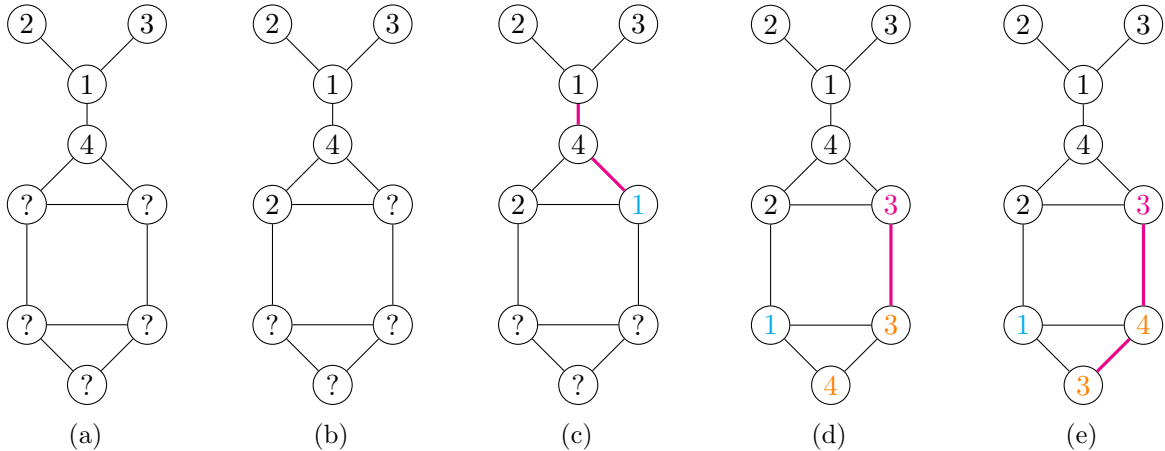

\paragraph{Coloring of the neighborhood.} Assume, that the neighborhood for a vertex \(v\) is \(\mathcal{N}(v) = \{w_1, w_2, w_3\}\). Then the coloring of such neighborhood is \(\chi(\mathcal{N}(v)) = \{\chi(w_1), \chi(w_2), \chi(w_3)\}\). As \(\chi\) is a distance 2 coloring, the colors in the set \(\chi(\mathcal{N}(v))\) are unique.

\begin{lemma}\label{g:lem:original_not_triangle}
Let \(v\) be an original vertex with color \(\chi(v)\) with \(\mathcal{N}(v) = \{w_1, w_2, w_3\}\). Then in the \(\gn_r(v)\) there is no such 3-tuple of vertices \(u_1, u_2, u_3\), that \(\{u_1, u_2\}, \{u_1, u_3\}, \{u_2, u_3\} \in E(G)\) and \(\chi(u_1) = \chi(v)\). Or in other words, a vertex that is part of a triangle (cycle of length 3) cannot be labeled with the same label as an original vertex.
\end{lemma}

\begin{proof}
    Assume by contradiction that such a 3-tuple exists. Then by 2nd assumption of \cref{thm:gadgeting} \(\chi(u_2) \in \chi(\mathcal{N}(v))\) and \(\chi(u_3) \in \chi(\mathcal{N}(v))\).

    W.l.o.g., assume \(\chi(u_2) = \chi(w_2)\) and \(\chi(u_3) = \chi(w_3)\). However, \(\{u_2, u_3\} \in E(G)\) which by 2nd assumption of \cref{thm:gadgeting} implies that there exists a vertex \(u \in \mathcal{N}(w_2)\) such that \(\chi(u) = \chi(w_3)\), which contradicts \cref{lem:outer_disconnected}. Thus, this 3-tuple doesn't exist.
\end{proof}

\begin{corollary}\label{g:cor:original_is_original}
Let \(v\) be an original vertex with color \(\chi(v)\). Then for every inner or outer vertex \(u\) in \(\gn_r(v)\), \(\chi(u) \neq \chi(v)\), as every inner and outer vertex is part of a triangle. Symmetrically, if we have an inner or outer vertex \(u\) with color \(\chi(u)\), then for every original vertex \(v\) in the \(r\)-neighborhood \(\chi(v) \neq \chi(u)\).
\end{corollary}

\begin{corollary}\label{g:cor:outer_is_outer}
Let \(w\) be an outer vertex with color \(\chi(w)\) neighboring an original vertex \(v\), then there is no such inner vertex \(u\) in \(\gn_r(v)\), such that \(\chi(u) = \chi(w)\). As otherwise by the second assumption of \cref{thm:gadgeting}, there would exist a vertex \(u' \in \mathcal{N}(u)\) such that \(\chi(u) =\chi(v)\). By distance definition, \(u' \in \gn_r(v)\). As vertex \(u\) is an inner vertex, then \(u'\) might be an outer or inner vertex, which by \cref{g:cor:original_is_original} cannot have the same color as an original vertex \(v\).
\end{corollary}

\paragraph{Gadget sequence.}
Gadget sequence is a sequence of gadgets between two original vertices (not including the ends)

\paragraph{Inner edge distance.}
We will define \(\edist\) (inner edge distance) for vertices that belong to the same gadget sequence or are original vertices at the ends of such a sequence. Then this is the distance function defined on the induced subgraph on the described vertices.

\begin{lemma}\label{g:lem:gadget_seq_correct}
    For every pair of vertices \((u, v)\), belonging to the same sequence of gadgets, \(\chi(u) \neq \chi(v)\).
\end{lemma}

\begin{proof}
    Consider a sequence of gadgets starting from an original vertex \(v\). Assume, for the sake of contradiction, that there is at least one pair of vertices \((u,\ u')\) for which the labeling is the same.

    Firstly, notice that \(\dist(u, u') \leq \edist(u, u')\) as \(\edist\) does not account for some edges.

    Thus, \(\edist(u, u')\) must be bigger than 2, as otherwise \(\dist(u, u') \leq 2\) and \(\chi\) would not be a distance 2 coloring, contradicting the first assumption of \cref{thm:gadgeting}. By construction of a gadget sequence this implies \(|\edist(u, v) - \edist(u', v)| > 2\).

    Let \(k\) be the smallest number such that there exist \(u, u'\) with \(max(\edist(u', v) , \edist(u, v)) = k\) and \(\chi(u) = \chi(u')\). W.l.o.g., assume that \(\edist(u, v) = k\), or in other words \(u\) is a further vertex.  Now, let \(w \in \mathcal{N}(u)\), such that \(\edist(w, v) < \edist(u, v)\). Such vertex exists as \(\edist(u,v) \geq \edist(u',v) + 2 \geq 0 + 2 = 2\). Then, by the second assumption of \cref{thm:gadgeting} vertex \(u'\) must have a neighbor \(w'\), such that \(\chi(w') = \chi(w)\). Note that \(w \neq w'\), as \(\edist(u, u') > 2\) and thus \(u\), \(u'\) cannot have a common neighbor.

    Then, \(\dist(w', v) \leq \dist(u', v) + 1 < \dist(u, v) = k\), \(\dist(w, v) < \dist(u, v) =k\) and \(\chi(w) = \chi(w')\). Thus, \(k\) is not the smallest distance on which a pair with the same labeling exists, contradicting the assumption about \(k\). Thus, such pair does not exist and every vertex within a gadget sequence has a unique color.
\end{proof}

\begin{lemma}\label{g:lem:outer_unique}
    Let \(v\) be an outer vertex with color \(\chi(v)\), neighboring an original vertex \(w\). Then, for every vertex \(v' \neq v\) in \(\gn_r(w):\) \(\chi(v') \neq \chi(v)\).
\end{lemma}

\begin{proof}
    Assume by contradiction that there are two sequences that start from outer vertices \(v\) and \(v'\), such that \(\chi(v) = \chi(v')\).

    First, we will introduce some notations. We denote by \(S_k\) and \(S'_k\) the set of vertices with the inner edge distance \(k\) from \(v\) and \(v'\), respectively, in the corresponding gadget sequences. Let \(u, u'\) be the outer vertices on the other side of the gadget sequences from \(v\) and \(v'\), respectively. Finally, let \(l = \edist(v, u)\) and \(l' = \edist(v', u')\).

    Now, we will prove that \(\forall k \leq min(l, l'):\  \chi(S_k) = \chi(S'_k)\).

    Assume for the sake of contradiction, that \(k\) is the minimal number such that exist \(\chi(S_k) \neq \chi(S'_k)\). By assumption, \(\chi(v) = \chi(v')\), therefore \(k \geq 1\).

    If \(k = 1\), then \(\chi(S_{1})\neq \chi(S'_{1})\). However, \(\chi(w)\cup\chi(S_{1}) = \chi(N(v)) = \chi(N(v')) = \chi(w')\cup\chi(S'_{1})\). As \(\chi(S_{1})\neq \chi(S'_{1})\) it implies that \(\chi(w') \in \chi(S_1)\). But by \cref{g:cor:original_is_original} an inner vertex cannot share a color with original inside \(\gn(v)\). Thus, \(k > 1\).

    Because minimality of \(k\), \(\chi(S_{k-1}) = \chi(S'_{k - 1})\) or in other words \(\forall w \in S_{k - 1} \exists w' \in S'_{k - 1}\) such that \(\chi(w) = \chi(w')\). Thus, by the second assumption of \cref{thm:gadgeting} \(\chi(\bigcup_{w \in S_k} N(w)) = \chi(\bigcup_{w' \in S'_k} N(w'))\). By construction, it implies \(\chi(S_{k - 2}) \cup \chi(S_{k - 1}) \cup \chi(S_{k}) = \chi(S'_{k - 2}) \cup \chi(S'_{k - 1}) \cup \chi(S'_{k})\). From \cref{g:lem:gadget_seq_correct}, we know that every vertex in a gadget sequence has a unique color within the sequence. By assumption \(\chi(S_k) \neq \chi(S'_k)\), thus \(\chi(S_{k - 2}) \cup \chi(S_{k - 1}) \neq \chi(S'_{k - 2}) \cup \chi(S'_{k - 1})\). However, this means that \(\chi(S_{k - 1}) \neq \chi(S'_{k - 1})\) or \(\chi(S_{k - 2}) \neq \chi(S'_{k - 2})\) and therefore \(k\) is not minimal. Thus, such \(k\) does not exist and \(\chi(S_{min(l, l')}) = \chi(S'_{min(l, l')})\). However, at least one of the gadget sequences has an outer vertex laying on the distance \(min(l, l')\) and by \cref{g:cor:outer_is_outer} the color of an outer vertex cannot be reused as the color of an inner vertex. Thus,  \(l\) must be equal to \(l'\), which is impossible as in the correct expansion there are no 2 sequences that have the same length in the $r$-gneighborhood.

    Therefore, we have reached a final contradiction and the color of the outer vertex is unique for an \(r\)-gneighborhood.
\end{proof}

\begin{corollary}
	Let \(v\) be an original vertex with color \(\chi(v)\). Then, for every vertex \(u \neq v\) in \(\gn_r(v):\)  \(\chi(u) \neq \chi(v)\). Assume by contradiction there is such \(u\) that \(\chi(u) = \chi(v)\). By \cref{g:cor:original_is_original} \(u\) must be an original vertex, so \(u\) is connected to an outer vertex \(w\), such that \(gdist(v, w) \leq gdist(v, u)\). Then by 2nd assumption of \cref{thm:gadgeting} there is \(u' \in \mathcal{N}(u)\), such that \(\chi(u') = \chi(w)\), which contradicts \cref{g:lem:outer_unique}. Thus, such \(u\) cannot exist.
\end{corollary}

\begin{corollary}
Let \(v\) be an inner vertex with color \(\chi(v)\). Then, for every vertex \(u \neq v\) in the \(r\)-gneighborhood, \(\chi(u) \neq \chi(v)\). This is because the color of the inner vertex is defined by its place in a gadget sequence between two outer vertices, which have unique colors in the \(r\)-neighborhood.
\end{corollary}
\begin{corollary}\label{g:cor:correct_gadgeting}
    \(\chi\) is a correct gdistance 2r original coloring. Thus, \cref{thm:gadgeting} is proved.
\end{corollary}

\section{From PN-checkable to RE formalism}\label{sec:re-formalism}

In this section we reduce an arbitrary always-solvable PN-checkable LCL on 3-regular graphs \(D =(\mathcal{G}_3, \outlabels_D, \mathcal{C}_D,r_D)\) to a round eliminator problem \(E = (\Sigma^{out}_{E}, \mathcal{C}_V, \mathcal{C}_E)\).
\paragraph{Rooted trees of radius \(r\).}
We say that a rooted tree \((T,v)\) has radius \(r\) if \(T^r_T(v)\cong (T,v)\).

\paragraph{Class of radius \(r\) rooted trees.}
Let \(\mathcal{T}_r = \{(T,v)\;:\; (T,v)\text{ is a rooted tree of radius } r\}.\)

\paragraph{Setup.}
Let \(C_D\) be the set of PN-views accepted for problem \(D\).

\paragraph{Neighbor enumeration.}
We define auxiliary functions
\[
\gamma : ((T,v),\sigma)\times \mathcal{N}(v)\to[3]
\quad\text{and}\quad
\zeta : ((T,v),\sigma)\times[3]\to\mathcal{N}(v),
\]
where the mapping \(\gamma(((T,v),\sigma),\cdot)\) is injective.
Moreover,
\[
\zeta\bigl(((T,v),\sigma),\gamma(((T,v),\sigma),w)\bigr)=w
\quad
\text{for all } w\in\mathcal{N}(v).
\]
\noindent
In other words, \(\gamma\) provides an enumeration for neighbors of \(v\) in \((T, v)\) and \(\zeta\) provides a neighbor of \(v\) by its number.

\paragraph{Directional and pruned subtrees.}
For \((T, v) \in \mathcal{T}_r\) and \(w\in\mathcal{N}(v)\), let \((T,v \to w)\) denote the rooted
labeled subtree of depth \(r-1\) obtained by re-rooting at \(w\) and
restricting  to vertices that lie
in the component of \((T, v)\) containing \(w\).
We write its labeled version as \(((T, v\to w),\sigma)\) and its PN-view version as \(T^r_G(v \to w)\).

Let \((T,v \setminus w)\) be the rooted labeled subtree of depth \(r-1\)
obtained by removing the branch containing \(w\).
We write its labeled version as \(((T, v\setminus w),\sigma)\) and its PN-view version as \(T^r_G(v \setminus w)\).

\begin{remark}
By the definition of PN-views, for any adjacent vertices \(u\) and \(v\),
\[
	(T^r_G(v \to u), \sigma)
\cong
(T^r_G(u \setminus v), \sigma).
\]
\end{remark}

\paragraph{Compatibility indicator.}
We define the compatibility indicator
\[
I_{\mathrm{fit}} : (C_D\times[3])^2 \to \{0,1\}.
\]
\noindent
For \((((T,v),\sigma),x),\;
(((T',v'),\sigma'),x')
\in C_D\times[3],\)
let
\[
u = \zeta(((T,v),\sigma),x) \in \mathcal{N}(v),
\quad
u' = \zeta(((T',v'),\sigma'),x') \in \mathcal{N}(v').
\]
\noindent
Then we set
\[
I_{\mathrm{fit}}
\bigl(
(((T,v),\sigma),x),
(((T',v'),\sigma'),x')
\bigr)=1
\]
if and only if
\[
(T_{v\to u},\,\sigma) \cong (T'_{v'\setminus u'},\,\sigma')
\quad\text{and}\quad
(T_{v\setminus u},\,\sigma) \cong (T'_{v'\to u'},\,\sigma').
\]
\noindent
Otherwise,
\[
I_{\mathrm{fit}}
\bigl(
(((T,v),\sigma),x),
(((T',v'),\sigma'),x')
\bigr)=0.
\]

\begin{definition}[Problem E]\label{def:problem:e} Problem E is an RE-problem, such that
\begin{itemize}
    \item \(\mathcal{G}_E\) is a family of 3-regular graphs.
    \item \(\outlabels_E = C_D \times [3]\).
    \item \(\mathcal{C}_V \coloneqq \{[(a, 1), (a, 2), (a, 3)] \mid a \in C_D\}\), a condition on nodes.
    \item \(\mathcal{C}_E \coloneqq \{[(a, x_1), (b, x_2)] \mid I_{\mathrm{fit}}((a, x_1), (b, x_2))\}\), a condition on edges.
    \item Let \(t : V(G) \to C_D\) be the function that assigns to each vertex \(v \in V(G)\) the element \(a \in C_D\) associated with it in the \(\mathcal{C}_V\) condition. We write \(t(v) = (t_T(v), t_\sigma(v))\).
    \item \(l: V(G) \to \outlabels_D\) and \(l(v) = t_\sigma(v)(v)\). So the labeling of the node corresponds to the labeling of the root node in its output.
\end{itemize}
\end{definition}

\begin{lemma}\label{lem:t_good}
	If for every vertex \(\mathcal{C}_V\) holds and for every edge \(\mathcal{C}_E\) holds, then for a vertex \(v\in V(G)\): \((T_{G}^r(v), l) \cong t(v)\). 
\end{lemma}
\begin{proof}
	We will be proving this statement by induction. We will also introduce an additional notation. 
		Labeled trees \(((T, v), \sigma)\) and \(((T', v'), \sigma')\) are \(k\)-isomorphic (denoted as \(\cong_k\)) iff \[((T, v)[w \in V((T, v))| dist(v, w)\leq k], \sigma) \cong ((T', v')[w \in V((T', v'))| dist(v', w)\leq k], \sigma'),\] where \(G[X]\) is notation for the induced subgraph of \(G\) on \(X\).
	
		Now, let us notice that for every node \(v\): \((T^r_{G}(v), l) \cong_0 t(v)\), as per definition \(l(v) = t_{\sigma}(v)(v)\). Thus the base of induction is proved. 

	Assume, that for every vertex \(v\), \((T^r_{G}(v), l) \cong_k t(v)\) for some \(0 \leq k \leq r - 1\). Let us select an arbitrary neighbor \(u \in \mathcal{N}(v)\). Then, from definitions \((T^{r}_{G}(v\to u), l) \cong (T^{r}_{G}(u\setminus v), l)\);  from \((T^r_{G}(u), l) \cong_k t(u)\) follows \((T^{r}_{G}(u \setminus v), l) \cong_k t(u \setminus v)\).
	From passive constraint \(t(u \setminus  v) \cong t(u \to v)\).
	Thus, \((T^r_{G}(v \to u), l) \cong t(v \to u)\).
		As \(u\) was selected arbitrarily, the same logic follows for the remaining neighbors of \(v\).
	As \((T^r_{G}(v), l)\) and \(t(v)\) are 3-regular trees of depth \(r\), with same labeling of the roots
	and each subtree being \(k\)-isomorphic, then \((T^r_{G}(v), l) \cong_{k + 1} t(v)\).
	Induction step proved, so \((T^r_{G}(v), l) \cong t(v)\). 
\end{proof}

\begin{lemma}[Soundness of the lift from D to E]\label{lem:sound:de}
Let $G_D$ be an instance of problem $D$. Let $G_E \cong G_D$ be an instance of problem \(E\). Let \(\sigma_E\) be a labelling on half-edges of \(G_E\)  such that all conditions on active and passive nodes are met.
	Define \(\ssigma{D}{out}\) as \(l\) in problem \(D\) description, then \(\ssigma{D}{out}\) is legal for problem \(D\). 
\end{lemma}
\begin{proof}
		By \cref{lem:t_good}, \(\forall v\in V(G_E)\): \((T_{G}^r(v), l) \cong t(v)\). Let \(u\) be an arbitrary vertex in \(V(G_D)\) and let \(v \in V(G_E)\) be its match,
		then by construction \((T_{G_E}^r(v), l) \cong (T_{G_D}^r(u), \ssigma{D}{out})\). 
		By the condition for active nodes, \(t(v) \in C_D\), implying \((T_{G_D}^r(u), \ssigma{D}{out}) \in C_D\), thus \(\ssigma{D}{out}\) is legal. 
\end{proof}

\begin{lemma}[Soundness of the lift from E to D]\label{lem:sound:ed}
	Let $G_E$ be an instance of problem $E$. Let $G_D \cong G_E$ be an instance of problem \(D\) and let \(\ssigma{D}{out}\) be legal for \(G_D\). 
		Define \(\sigma_{E}\) as a labeling on half-edges, such that
	\(\sigma_{E}(v, \{v, u\}) = ((T^r_{G_D}(v), \ssigma{D}{out}), \gamma((T^r_{G_D}(v), \ssigma{D}{out}), u))\).

	Then, \(\sigma_{E}\) is legal for problem \(E\).
\end{lemma}
\begin{proof}
		By definition of \(\sigma_{E}\) and \(\gamma\), the constraints on nodes are fulfilled. Since \((T^r_G(v \to u), \sigma) \cong (T^r_G(v \setminus u), \sigma)\), the constraints on edges are also fulfilled. Thus \(\sigma_{E}\) is legal for problem \(E\). 
\end{proof}

\begin{lemma}
	Let $\instance_D = (G_D)$ be an instance of problem $D$. Let $G_E\cong G_D$ be an instance of problem $E$. Let $\outcome_E = \{(\sigma_E^{\oupt,i},p_i)\}_{i\in I}$ be an outcome that solves $E$ on $G_E$ with probability at least $q$. Define $\outcome_D$ by sampling $\{(\sigma_E^{\oupt,i},p_i)\}_{i\in I}$ from $\outcome_E$ and applying the $\LIFT$. If the lift is sound, then $\outcome_D$ solves $D$ with probability at least $q$.
\end{lemma}
\begin{lemma}
		Let $\instance_E = (G_E)$ be an instance of problem $E$. Let $G_D\cong G_E$ be an instance of problem $D$. Let $\outcome_D = \{(\sigma_D^{\oupt,i},p_i)\}_{i\in I}$ be an outcome that solves $D$ on $G_D$ with probability at least $q$. Define $\outcome_E$ by sampling $\{(\sigma_D^{\oupt,i},p_i)\}_{i\in I}$ from $\outcome_D$ and applying the $\LIFT$. If the lift is sound, then $\outcome_E$ solves $E$ with probability at least $q$.
\end{lemma}
The proof of the lemma follows the same steps as the proofs for \cref{lemma:prob_a_to_b,lemma:prob_b_to_a} and is thus omitted.

\begin{lemma}\label{lemma:simulation_d_e}
If problem $D$ admits an optimal algorithm with locality $\Theta(T) = \Omega(\diamondsuit(\mathcal{M}))$ for a model \(\mathcal{M}\), then problem $E$ admits an algorithm with locality $\Theta(T)$.
\end{lemma}
\begin{proof}
First, let us prove that such \(\alg{E}\) with \(\Theta(T)\) exists. Such an algorithm can consist of two steps. First, the nodes proceed with \(\alg{D}\), which will have \(\Theta(T)\) locality.
The final output depends only on \((T^r_{G}(v), \sigma_{out})\); comparing it to \(C_D\) and outputting \(((T^r_{G}(v), \sigma_{out}), x)\) with correct enumeration of \(x\) on the half-edges has constant locality \(O(1)\). By \cref{lem:sound:ed}, such an algorithm produces a valid solution.

Now assume, that there is a faster algorithm \(\alg{E}\) with locality \(o(T)\), then we can run this algorithm. However, by \cref{lem:sound:de} this will also solve problem D, if we will take the \(l\) as the output, which takes only an internal computation. Then \(\alg{D}\) was not optimal. A contradiction.
\end{proof}

\section{Final proof}\label{sec:final}
In this section, we finalize the main theorem of this paper and provide a corresponding proof.

\begin{theorem}
		The complexity landscape for LCL problems and RE problems is the same for any model \(\mathcal{M}\) above its symmetry-breaking region \(\diamondsuit(\mathcal{M})\).

		Or in other words, if an always-solvable LCL problem \(A\) admits an optimal algorithm with locality \(\Theta(T) = \Omega(\diamondsuit(\MM))\) for a model \(\MM\), then there exists an RE problem \(E\), which admits an algorithm with locality \(\Theta(T)\). And vice versa.
\end{theorem}

\begin{proof}
		By \cref{lemma:simulation_a_b}, for any always-solvable LCL problem \(A\) which admits an optimal algorithm with locality \(\Theta(T) = \Omega(\diamondsuit(\MM))\), we can construct an always-solvable LCL problem \(B\) on 3-regular graphs without inputs that admits an algorithm with locality \(\Theta(T)\). Following the same logic through \cref{lemma:simulation_b_d,lemma:simulation_d_e}, we can construct an RE problem \(E\), which admits an algorithm with locality \(\Theta(T)\).

		Now, for the other direction, it is trivial that problems in the RE formalism can be easily expressed as LCL problems. Thus, both formalisms have the same complexity landscape above the symmetry-breaking region. 
\end{proof}

\section*{Acknowledgements.}
We would like to thank Alkida Balliu, Dennis Olivetti, Sebastien Brandt and Henrik Lievonen for many fruitful discussions. This work was supported in part by the Research Council of Finland, Grants 363558 and 359104.  
\bibliography{da,extra}

\end{document}